\pdfoutput=1
\def \VersionLong {}
\def \VersionFinal {}

\ifdefined \VersionLong
	\newcommand{\LongVersion}[1]{\ifdefined\VersionWithComments{\color{red!40!black}#1}\else#1\fi}
	\newcommand{\ShortVersion}[1]{\ifdefined\VersionWithComments{\color{black!40}#1}\fi}
\else
	\newcommand{\LongVersion}[1]{\ifdefined\VersionWithComments{\color{black!40}#1}\fi}
	\newcommand{\ShortVersion}[1]{\ifdefined\VersionWithComments{\color{red!40!black}#1}\else#1\fi}
\fi

\ifdefined \VersionFinal
	\newcommand{\FinalVersion}[1]{#1}
	\newcommand{\AnonymousVersion}[1]{}
\else
	\newcommand{\FinalVersion}[1]{}
	\newcommand{\AnonymousVersion}[1]{#1}
\fi

\FinalVersion{\documentclass[sigconf,10pt]{acmart}}
\AnonymousVersion{\documentclass[sigconf,review,screen,anonymous,10pt]{acmart}}

\AtBeginDocument{%
  \providecommand\BibTeX{{%
    \normalfont B\kern-0.5em{\scshape i\kern-0.25em b}\kern-0.8em\TeX}}}

\copyrightyear{2020}
\acmYear{2020}
\setcopyright{acmcopyright}
\acmConference[HSCC '20]{23rd ACM International Conference on Hybrid Systems: Computation and Control}{April 22--24, 2020}{Sydney, NSW, Australia}
\acmBooktitle{23rd ACM International Conference on Hybrid Systems: Computation and Control (HSCC '20), April 22--24, 2020, Sydney, NSW, Australia}
\acmPrice{15.00}
\acmDOI{10.1145/3365365.3382193}
\acmISBN{978-1-4503-7018-9/20/04}

\usepackage[ruled,lined,linesnumbered,noend]{algorithm2e}
\SetKwInOut{Input}{input}
\SetKwInOut{Output}{output}
\SetKw{KwPush}{push}
\SetKw{KwPop}{pop}
\SetKw{KwFrom}{from}
\SetKw{KwCompute}{compute}
\SetKw{KwBreak}{break}
\SetKwFor{Until}{until}{do}{end}

\usepackage{subcaption}
\usepackage{paralist} 
\usepackage{xspace}

\newenvironment{ienumeration}
	{\ifdefined\VersionLong\begin{enumerate}\else\begin{inparaenum}[\itshape i\upshape)]\fi}
	{\ifdefined\VersionLong\end{enumerate}\else\end{inparaenum}\fi}

\usepackage{amsmath} 
\usepackage{multirow}
\usepackage{stmaryrd}

\usepackage[misc,geometry]{ifsym} 

\ifdefined\VersionWithComments
	\usepackage{draftwatermark}
	\SetWatermarkText{draft}
	\SetWatermarkScale{3}
	\SetWatermarkColor[gray]{0.9}
\fi

\usepackage{xcolor}
\usepackage{colortbl}
\definecolor{darkblue}{rgb}{0.0,0.0,0.6}
\definecolor{darkgreen}{rgb}{0, 0.5, 0}
\definecolor{darkpurple}{rgb}{0.7, 0, 0.7}
\definecolor{darkblue}{rgb}{0, 0, 0.7}

\usepackage[capitalise,english,nameinlink]{cleveref} 
\crefname{line}{\text{line}}{\text{lines}} 
\crefname{assumption}{\text{Assumption}}{\text{Assumptions}} 

\usepackage{tikz}
\usetikzlibrary{arrows,automata,positioning}
\tikzstyle{every node}=[initial text=]
\tikzstyle{location}=[rectangle, rounded corners, minimum size=12pt, draw=black, fill=blue!10, inner sep=2pt]
\tikzstyle{final}=[double]
\tikzstyle{accepting}=[final]
\tikzstyle{PTPMOPT}=[,dashed,color=red,semithick]

\definecolor{coloract}{rgb}{0.50, 0.70, 0.30}
\definecolor{colorclock}{rgb}{0.4, 0.4, 1}
\definecolor{colorconst}{rgb}{0.50, 0.20, 0.00}
\definecolor{colordisc}{rgb}{1, 0, 1}
\definecolor{colorloc}{rgb}{0.4, 0.4, 0.65}
\definecolor{colorparam}{rgb}{1, 0.6, 0.0}

\newif\iftikzgnuplot
\tikzgnuplottrue 
\usepackage{gnuplot-lua-tikz}
\usepackage{pgfplotstable}
\pgfplotsset{compat=1.12}
\usepackage{booktabs}

\usetikzlibrary{automata,positioning,matrix,shapes.callouts}
\tikzset{
region/.style={
rectangle,
rounded corners,
draw=black,very thick
},
accepting/.style={double distance=2pt}
}

\newcommand{\N}{{\mathbb{N}}}

\newcommand{\R}{{\mathbb{R}}}

\newcommand{\ttrue}{\mathrm{t{\kern-1.5pt}t}}
\newcommand{\ffalse}{\mathrm{f{\kern-1.5pt}f}}

\newcommand{\powerset}[1]{\mathcal{P}({#1})}
\newcommand{\setdiff}{\mathrel{\triangle}}

\makeatletter
\newcommand{\figcaption}[1]{\def\@captype{figure}\caption{#1}}
\newcommand{\tblcaption}[1]{\def\@captype{table}\caption{#1}}
\makeatother

\usepackage{arydshln}
\usepackage{enumitem}
\setlistdepth{20}
\renewlist{itemize}{itemize}{20}
\setlist[itemize]{label=\textbullet}

\newcommand{\defProblem}[3]
{%
\noindent\fcolorbox{black}{blue!15}{
	\begin{minipage}{.95\columnwidth}
		\textbf{#1:}\\
		\textsc{Input}: #2\\
		\textsc{Problem}: #3
	\end{minipage}
}
}

\AtEndPreamble{
\theoremstyle{acmplain}
\newtheorem{mytheorem}{Theorem}[section]

\newtheorem{mycorollary}[mytheorem]{Corollary}
\theoremstyle{acmdefinition}

\newtheorem{mydefinition}[mytheorem]{Definition}

}

\ifdefined\VersionWithComments
	\usepackage[colorinlistoftodos,textsize=footnotesize]{todonotes}
\else
	\usepackage[disable]{todonotes}
\fi
\newcommand{\gennote}[3]{\todo[inline,linecolor=#2,backgroundcolor=#2!25,bordercolor=#2]{#3: #1}}

\newcommand{\mw}[1]{\gennote{#1}{orange}{MW}}
\newcommand{\instructions}[1]{{\gennote{\bfseries #1}{red}{Instructions}}}

\newcommand{\abs}[1]{|#1|}

\newcommand{\propTrace}{\pi}
\newcommand{\PropTrace}{(\powerset{\AP})^{\infty}}
\newcommand{\FinPropTrace}{(\powerset{\AP})^{*}}
\newcommand{\InfPropTrace}{(\powerset{\AP})^{\omega}}

\newcommand{\IVar}{X}

\newcommand{\OVar}{Y}
\newcommand{\ovar}{y}
\newcommand{\signal}{\sigma}
\newcommand{\dval}{u}

\newcommand{\signalWithInside}{\signal=\dval_0,\dval_1,\cdots,\dval_{n-1}}
\newcommand{\signalWithInfInside}{\signal=\dval_0,\dval_1,\cdots}

\newcommand{\FinISignal}{(\R^\IVar)^{*}}

\newcommand{\OSignal}{(\R^\OVar)^{\infty}}
\newcommand{\FinOSignal}{(\R^\OVar)^{*}}
\newcommand{\InfOSignal}{(\R^\OVar)^{\omega}}

\newcommand{\stlFml}{\varphi}
\newcommand{\ltlFml}{\psi}
\newcommand{\satisfy}[3]{(#1,#2) \models #3}

\newcommand{\notsatisfy}[3]{\left(#1,#2 \right) \not\models #3}
\newcommand{\sem}[1]{\llbracket #1 \rrbracket} 
\newcommand{\supSem}[1]{\llbracket #1 \rrbracket_{\Diamond}} 
\newcommand{\infSem}[1]{\llbracket #1 \rrbracket_{\square}}

\newcommand{\prop}{p}
\newcommand{\AP}{\mathbf{AP}}
\newcommand{\UntilOp}[1]{\mathbin{\mathcal{U}_{#1}}}

\newcommand{\DiaOp}[1]{\Diamond_{#1}}
\newcommand{\BoxOp}[1]{\square_{#1}}

\newcommand{\NextOp}{\mathcal{X}}
\newcommand{\interval}[1][i,j]{[#1)}
\newcommand{\constant}{c}
\newcommand{\robust}[3]{{ \rho(#1,#2,#3) }}
\newcommand{\RoSI}[3]{\mathrm{RoSI}(#1, #2, #3)}
\newcommand{\finRobust}[3]{{ [\rho](#1,#2,#3) }}

\newcommand{\word}{w}\newcommand{\A}{\mathcal{A}}
\newcommand{\Alphabet}{\Sigma}
\newcommand{\OAlphabet}{\Gamma}
\newcommand{\lts}{\mathcal{M}}
\newcommand{\loc}{l}
\newcommand{\Loc}{L}
\newcommand{\initLoc}{\loc_0}

\newcommand{\Lg}{\mathcal{L}}
\newcommand{\LgFin}{\mathcal{L}^{\mathrm{fin}}}
\newcommand{\LgIn}{\Lg_{\mathrm{in}}}
\newcommand{\LgOut}{\Lg_{\mathrm{out}}}
\newcommand{\inChar}{a}
\newcommand{\outChar}{b}
\newcommand{\ioChar}[1][]{(\inChar_{#1},\outChar_{#1})}

\newcommand{\Transition}{\Delta}

\newcommand{\ltsInside}{(\Loc, \initLoc, \Transition)}
\newcommand{\ltsWithInside}{\lts = \ltsInside}

\newcommand{\subseq}[3]{#1[#2,#3]}
\newcommand{\project}[2]{{\mathbf{pr}_{#2}(#1)}}
\newcommand{\actualModel}{\mathcal{M}}
\newcommand{\learnedModel}{\tilde{\mathcal{M}}}

\newcommand{\ICommand}{\Alphabet^*}
\newcommand{\icommand}{\iota}
\newcommand{\icommandInside}{\inChar_1,\inChar_2,\ldots,\inChar_n}
\newcommand{\icommandWithInside}{\icommand = \icommandInside}
\newcommand{\IMapper}{\mathcal{I}}
\newcommand{\OMapper}{\mathcal{O}}

\newcommand{\population}{I}
\newcommand{\initPopulation}{\mathrm{genPopul}}
\newcommand{\isTimeout}{\mathrm{isTimeout}}
\newcommand{\generateNextPopulation}{\mathrm{genNextPopulation}}
\newcommand{\unfalsified}{\mathit{notFalsified}}

\newcommand{\breach}{\textsc{Breach}}
\newcommand{\staliro}{\textsc{S-TaLiRo}}
\newcommand{\falcaun}{\textsf{FalCAuN}}

\newcommand{\ourtool}{\falcaun}
\newcommand{\ga}{\textsc{GA}}
\newcommand{\hc}{\textsc{HC}}
\newcommand{\random}{\textsc{Random}}

\ifdefined \VersionWithComments
 	\definecolor{colorok}{RGB}{80,80,150}
\else
	\definecolor{colorok}{RGB}{0,0,0}
\fi

\newcommand{\eg}{\textcolor{colorok}{e.g.,}\xspace}
\newcommand{\ie}{\textcolor{colorok}{i.e.,}\xspace}

\usepackage{graphicx}	
\usepackage{version}	
\usepackage{url}	
\begin{document}

\title{Falsification of Cyber-Physical Systems with Robustness-Guided Black-Box Checking}
\author{Masaki Waga}
\authornote{JSPS Research Fellow}
\orcid{0000-0001-9360-7490}
\affiliation{%
  \institution{National Institute of Informatics and the Graduate University for Advanced Studies}
  \streetaddress{2-1-2 Hitotsubashi}
  \city{Chiyoda}
  \state{Tokyo}
  \postcode{101-8430}
  \country{Japan}}
\email{mwaga@nii.ac.jp}

\renewcommand\shortauthors{Masaki, W.}

\begin{abstract}
For exhaustive formal verification, industrial-scale \emph{cyber-physical systems (CPSs)} are often too large and complex, and lightweight alternatives (\eg{} monitoring and testing) have attracted the attention of both industrial practitioners and academic researchers.  \emph{Falsification} is one popular testing method of CPSs utilizing \emph{stochastic optimization}. In state-of-the-art falsification methods, the result of the previous falsification trials is discarded, and we always try to falsify without any prior knowledge. To concisely memorize such prior information on the CPS model and exploit it, we employ \emph{Black-box checking (BBC)}, which is a combination of \emph{automata learning} and \emph{model checking}. Moreover, we enhance BBC using the \emph{robust semantics} of STL formulas, which is the essential gadget in falsification. Our experiment results suggest that our robustness-guided BBC outperforms a state-of-the-art falsification tool.
\end{abstract}

%
%


%
%


\maketitle

\ifdefined \VersionWithComments
	\textcolor{red}{\textbf{This is the version with comments. To disable comments, comment out line~3 in the \LaTeX{} source.}}
\fi

\instructions{HSCC: Maximum 10 pages, 10pt font, two-column ACM format.}


\section{Introduction}\label{section:introduction}

\paragraph{Falsification of cyber-physical systems}
Due to their safety-critical nature, safety assurance of \emph{cyber-physical systems (CPSs)} is a vital problem.
For \emph{exhaustive} formal verification, \eg{} reachability analysis, industrial-scale cyber-physical systems (CPSs) are often too large and complex. Therefore \emph{non-exhaustive} but lightweight alternatives (\eg{} \emph{monitoring} and \emph{black-box testing}) have attracted the attention of both industrial practitioners and academic researchers. 
 \emph{Optimization-based falsification} is one of the \emph{search-based testing} methods to find bugs in CPSs, and many algorithms~\cite{DBLP:conf/hybrid/NghiemSFIGP10,DBLP:conf/case/DokhanchiYHF17,DBLP:conf/nfm/DreossiDS17,DBLP:journals/tcad/ZhangESAH18,DBLP:conf/adhs/YaghoubiF18,DBLP:conf/cav/ZhangHA19} have been studied. The problem is formulated as follows. 

\defProblem{The falsification problem}{a CPS model $\actualModel$ (given an input signal $\signal$, it returns an output signal $\actualModel(\signal)$) and a specification $\stlFml$ of\LongVersion{ the CPS model} $\actualModel$.}{Find a violating input signal $\signal$ such that the corresponding output signal $\actualModel(\signal)$ violates\LongVersion{ the specification} $\stlFml$\LongVersion{ \ie{} $\actualModel(\signal) \not\models \stlFml$}}

The technical essence of optimization-based falsification is to reduce CPS safety assurance to the \emph{simulation-based optimization} problem through the \emph{robust semantics}~\cite{DBLP:journals/tcs/FainekosP09} of \emph{signal temporal logic (STL)} formulas~\cite{DBLP:conf/formats/MalerN04}. The robust semantics of an STL formula shows a \emph{quantitative} satisfaction degree: if the robust semantics of an STL formula $\stlFml$ is negative, $\stlFml$ is violated. 
Thus, the falsification problem can be solved by minimizing the robust semantics of the given STL formula $\stlFml$ using an optimization technique, \eg{} covariance matrix adaptation evolution strategy (CMA-ES)~\cite{DBLP:conf/cec/AugerH05}, through simulations. 
The analysis of differential equations tends to be expensive, and falsification often finds a bug more efficiently than formal verification of CPSs such as reachability analysis\LongVersion{ of hybrid automata}.

Thanks to the robust semantics of STL, optimization-based falsification often falsifies an STL formula effectively even if it is hard for a random testing.
Falsification usually requires many simulations to find a violating input signal.
This can be a problem due to the simulation cost of CPSs.
A simulator of a self-driving car---involving the obstacles (\eg{} pedestrians and other cars) and road conditions as well as the ego car---typically runs in a speed that is more or less real-time. 
A single simulation of it thus would take several seconds, at least.
Thus, we want to reduce the number of the simulations.

\paragraph{Black-box checking}
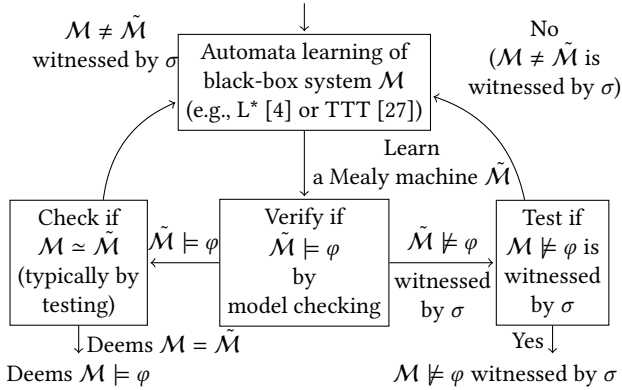
\begin{figure}[tbp]
 \centering
 \begin{tikzpicture}[shorten >=1pt,node distance=4cm,on grid,auto] 
  \small
  \node[initial above,rectangle,draw](learning)[align=center] {Automata learning of\\ black-box system $\actualModel$\\(\eg{} L*~\cite{DBLP:journals/iandc/Angluin87} or TTT~\cite{DBLP:conf/rv/IsbernerHS14})};
  \node[rectangle,draw,node distance=2.4cm](MC)[below=of learning, align=center] {Verify if\\ $\learnedModel \models \stlFml$ \\by \\model checking};
  \node[node distance=3.0cm,rectangle,draw](equivalence)[left=of MC,align=center] {Check if\\ $\actualModel \simeq \learnedModel$\\ (typically by\\ testing)};
  \node[node distance=1.5cm](maybesafe)[below=of equivalence,align=right] {Deems $\actualModel \models \stlFml$};
  \node[node distance=3.3cm,rectangle,draw](cexcheck)[right=of MC,align=center] {Test if\\ $\actualModel \not\models \stlFml$ is \\witnessed\\ by $\signal$};
  \node[node distance=1.5cm](falsified)[below=of cexcheck,align=center] {$\actualModel \not\models \stlFml$ witnessed by $\signal$\qquad\qquad};
 \path[->] 
 (learning) edge node[align=center] {Learn \\a Mealy machine $\learnedModel$} (MC)
 (MC) edge node[above] {$\learnedModel \models\stlFml$} (equivalence)
 (equivalence) edge[bend left] node[above=1.4,align=center,pos=0.1] {$\actualModel \neq \learnedModel$\\ witnessed by $\signal$} (learning)
 (equivalence) edge node[right,align=center] {Deems $\actualModel = \learnedModel$} (maybesafe)
 (MC) edge node[above] {$\learnedModel \not\models\stlFml$} node[below,align=center] {witnessed \\ by $\signal$} (cexcheck)
 (cexcheck) edge[bend right] node[above right,align=center,pos=0.8] {No\\ ($\actualModel \neq \learnedModel$ is\\ witnessed by $\signal$)} (learning)
 (cexcheck) edge node[left] {Yes} (falsified);
 \end{tikzpicture}
 \caption{A workflow of black-box checking~\cite{DBLP:journals/jalc/PeledVY02,DBLP:conf/dagstuhl/Meinke16}.}
 \label{fig:bbc}
\end{figure}

\emph{Black-box checking (BBC)}~\cite{DBLP:journals/jalc/PeledVY02} or \emph{learning-based testing (LBT)}~\cite{DBLP:conf/dagstuhl/Meinke16} is another testing method of black-box systems.
The speciality of BBC is the combination of \emph{automata learning}~\cite{DBLP:conf/sfm/2011} and \emph{model checking}~\cite{DBLP:books/daglib/0020348}.
As described in~\cite{DBLP:journals/jalc/PeledVY02},
an outline of BBC is shown in \cref{fig:bbc}.
Here, a black-box system $\actualModel\colon \ICommand \to \FinPropTrace$ 
is a function from a discrete input sequence $\icommand \in \ICommand$ to a sequence $\actualModel (\icommand) \in \FinPropTrace$ of the sets of atomic propositions satisfied at each time.
By automata learning, a \emph{Mealy machine} $\learnedModel$ is constructed from the previous simulation results (the top box of \cref{fig:bbc}).
The learned Mealy machine $\learnedModel$ is used to approximate the black-box system $\actualModel$.
By model checking, one checks if the learned Mealy machine $\learnedModel$ satisfies the given property $\stlFml$ (bottom center of \cref{fig:bbc}).
Since the behavior of the black-box system $\actualModel$ and the learned Mealy machine $\learnedModel$ can be different, their consistency is confirmed through additional simulations of $\actualModel$ (bottom left and right of \cref{fig:bbc}).
We note that 
the learned Mealy machine $\learnedModel$ is independent of the property $\stlFml$,
and we can use\LongVersion{ the learned Mealy machine} $\learnedModel$ for model checking of properties other than $\stlFml$.

Thanks to the soundness of \emph{conformance testing} used as \emph{equivalence testing} (left of \cref{fig:bbc}),
BBC can guarantee that the given black-box system certainly satisfies the given property~\cite{DBLP:journals/jalc/PeledVY02} although the soundness relies on additional assumptions on the black-box system (\eg{} the upper bound of the number of the states).
A recent survey~\cite{DBLP:conf/dagstuhl/HowarS16} reports that
at the early stage of the automata learning,
it is beneficial 
for the equivalence testing
to try to find a counterexample $\icommand \in \ICommand$ satisfying $\actualModel(\icommand) \neq \learnedModel(\icommand)$ instead of 
trying to prove the equivalence by conformance testing, \eg{} W-method~\cite{DBLP:journals/tse/Chow78} and Wp-method~\cite{DBLP:journals/tse/FujiwaraBKAG91}.
\emph{Random testing} is one typical choice of the equivalence testing other than conformance testing.
Random testing usually samples the inputs uniformly, and it is good at covering various inputs.
But, due to its uniform nature, random testing is not good at finding \emph{rare} counterexamples existing only in a small area of the input space.

\paragraph{Robustness-guided black-box checking}
Our contribution is to combine optimization-based falsification and BBC aiming at the improvement of both of them.
We enhance BBC by the robust semantics of STL, which is the essential gadget in\LongVersion{ optimization-based} falsification. We utilize BBC to solve the falsification problem.

As an improvement of BBC,
we employ the robust semantics of STL to enhance the search of a counterexample exploiting the following observation.
If the CPS $\actualModel$ violates the given STL formula $\stlFml$, but the learned Mealy machine $\learnedModel$ satisfies $\stlFml$,
there exists a discrete input $\icommand \in \ICommand$ such that 
the output $\actualModel(\icommand)$ of $\actualModel$ violates the specification $\stlFml$, and
we have $\actualModel(\icommand) \neq \learnedModel(\icommand)$.
Minimizing the robust semantics of $\stlFml$,
our equivalence testing of $\actualModel$ and $\learnedModel$ focuses on a subspace of the input space where a counterexample more likely exists.
To minimize the robust semantics,
we use, \eg{} hill climbing or \emph{genetic algorithms}~\cite{DBLP:books/daglib/0070933}.

As an improvement of optimization-based falsification, 
we aim at reducing the number of the simulations 
when we try to falsify a CPS over multiple STL formulas.
Multiple STL formulas are used in falsification, \eg{} because for one abstract requirement in engineers' minds, many STL formulas realize it, and we want to try some STL formulas out of them.
Through the automata learning in BBC, we reuse the knowledge on the CPS obtained when falsifying other STL formulas, and reduce the number of the simulations.
See \cref{section:related_works} for related works on model learning for falsification.

Another big problem of optimization-based falsification is that we can obtain very small information when we failed to falsify it.
Since BBC generates a learned Mealy machine $\learnedModel$ even if the given specifications are not falsified, we can potentially use it to explain why the BBC failed.

We note that the existing robust semantics, \eg{}~\cite{DBLP:journals/tcs/FainekosP09,DBLP:conf/formats/DonzeM10,DBLP:conf/cav/AkazakiH15}, are incompatible with the finite semantics of LTL in~\cite{DBLP:conf/cav/dAmorimR05}, which is implemented in LTSMin~\cite{DBLP:conf/tacas/KantLMPBD15}.
Although the novelty is limited,
we define and employ a suitable robust semantics of STL with a soundness and correctness theorem.

We implemented a prototypical tool \ourtool{} for robustness-guided BBC and compared its performance with: 
\begin{inparaenum}[\itshape i\upshape)]
 \item \breach{}, which is one of the state-of-the-art falsification tools; and
 \item a baseline BBC using random search for the equivalence testing.
\end{inparaenum}
Our experimental result suggests that 
\begin{ienumeration}
 \item on average, robustness-guided BBC using \emph{genetic algorithm} falsifies more properties \breach{} and the baseline BBC method; and
 \item robustness-guided BBC is much more scalable than \breach{} with respect to the number of the properties we try to falsify.
\end{ienumeration}

\paragraph{Contributions}

Our contributions are summarized as follows.

\begin{itemize}
 \item By combining optimization-based falsification and black-box checking (BBC),
       we proposed robustness-guided BBC to improve both of them.
 \item We implemented a prototypical tool \ourtool{} for robustness-guided BBC.
 \item Our experimental results show that our robustness-guided BBC outperforms baseline BBC and one of the state-of-the-art falsification algorithms.
\end{itemize}

\paragraph{Organization}
After reviewing some preliminaries in \cref{section:preliminaries}, we show the robust semantics of STL in a \emph{discrete-time} setting in \cref{section:discrete_stl}.
This semantics is compatible with the finite semantics of LTL in~\cite{DBLP:conf/cav/dAmorimR05}.
In \cref{section:bbc-cps}, we show how to enhance BBC by the robust semantics of STL, which is the main contribution of this paper.
We show our experimental evaluation in \cref{section:experiments}.
We review some related works in \cref{section:related_works}.
We conclude and show future works in \cref{section:conclusions_and_future_work}.

\section{Preliminaries}\label{section:preliminaries}

\paragraph{Notations}
For a set $X$, we denote its powerset by $\powerset{X}$.
We denote the empty sequence by $\varepsilon$.
For a set $X$, 
an infinite sequence $\overline{x} = x_0, x_1,\dots \in X^{\omega}$ of $X$, and
$i,j \in N$ satisfying $i \leq j$, 
we denote the subsequence $x_i,x_{i+1}, \dots,x_j \in X^{*}$ by $\subseq{\overline{x}}{i}{j}$.
For a set $X$, 
a finite sequence $\overline{x} \in X^{*}$ of $X$, and
an infinite sequence $\overline{x'} \in X^{\omega}$ of $X$,
we denote their concatenation by $\overline{x} \cdot \overline{x'}$.
For a set $X$ and its subsets $X', X'' \subseteq X$, we denote the \emph{symmetric difference} of $X'$ and $X''$ by
$X' \setdiff X'' = \{x  \in X \mid x \in X', x \not\in X''\} \cup \{x  \in X \mid x \not\in X', x \in X''\}$.
For a function $f\colon X\to Y$ and 
a finite sequence $\overline{x} = x_1,x_2,\dots,x_n \in X^{*}$, 
we let $\overline{f}\colon X^* \to Y^*$ as $\overline{f}(\overline{x}) = f(x_1),f(x_2),\dots,f(x_n)$.
For closed intervals $I_1=[a_1,b_1], I_2=[a_2,b_2]$ over $\R \cup \{\pm\infty\}$, 
we let
 $-I_1 = [-b_1, -a_1]$ and
 $\max(I_1,I_2) = [\max(a_1,a_2), \max(b_1,b_2)]$ and
 $\min(I_1,I_2) = [\min(a_1,a_2), \min(b_1,b_2)]$.

\subsection{LTL model checking}

\emph{Linear temporal logic (LTL)}~\cite{DBLP:conf/focs/Pnueli77} is a commonly used formalism to describe temporal behaviors of an infinite or finite sequence $\propTrace \in \PropTrace$ of a set $\propTrace_i \subseteq \AP$ of atomic propositions representing valuations of atomic propositions.

\begin{mydefinition}
 [linear temporal logic]
 \label{def:ltl}
 For the  set $\AP$ of the atomic propositions, the syntax of \emph{linear temporal logic (LTL)} is defined as \LongVersion{follows}\ShortVersion{$\ltlFml, \ltlFml' ::= \top \mid \prop \mid \neg \ltlFml \mid \ltlFml \lor \ltlFml' \mid \ltlFml \UntilOp{[i,j)} \ltlFml' \mid \NextOp \ltlFml$}, where $\prop \in \AP$ and $i,j \in \N \cup \{+\infty\}$ satisfying $i \leq j$.\footnote{In the standard definition of LTL, the interval $[i,j)$ in $\UntilOp{[i,j)}$ is always $[0,\infty)$ and it is omitted. We employ the current syntax to emphasize the similarity to STL. We note that this does not change the expressive power.}\LongVersion{
\[
 \ltlFml, \ltlFml' ::= \top \mid \prop \mid \neg \ltlFml \mid \ltlFml \lor \ltlFml' \mid \ltlFml \UntilOp{[i,j)} \ltlFml' \mid \NextOp \ltlFml
\]}

For an LTL formula $\ltlFml$, an \emph{infinite} sequence $\propTrace = \propTrace_0, \propTrace_1,\dots \in \InfPropTrace$ of subsets of atomic propositions, and $k \in \N$, 
we define the satisfaction relation $\satisfy{\propTrace}{k}{\ltlFml}$ as follows.
\begin{align*}
 \satisfy{\propTrace}{k}{\top} \qquad&
 \satisfy{\propTrace}{k}{\prop} \iff \prop \in \propTrace_k\\
 \satisfy{\propTrace}{k}{\neg \ltlFml} \iff& (\propTrace, k) \not\models \ltlFml\\
 \satisfy{\propTrace}{k}{\ltlFml \lor \ltlFml'} \iff& \satisfy{\propTrace}{k}{\ltlFml} \lor \satisfy{\propTrace}{k}{\ltlFml'}\\
 \satisfy{\propTrace}{k}{\NextOp{\ltlFml}} \iff& \satisfy{\propTrace}{k+1}{\ltlFml}\\
 \satisfy{\propTrace}{k}{\ltlFml \UntilOp{[i,j)} \ltlFml'} \iff&
 \exists l \in [k+i,k+j).\,\satisfy{\propTrace}{l}{\ltlFml'} \\
 \qquad\qquad\qquad\qquad\land& \forall m \in \{k,k+1,\dots,l\}.\, \satisfy{\propTrace}{m}{\ltlFml}
\end{align*}
\end{mydefinition}

We denote $\propTrace \models \ltlFml$ if we have $\satisfy{\propTrace}{0}{\ltlFml}$.
An LTL formula $\ltlFml$ is \emph{safety} if 
for any infinite sequence $\propTrace \in \InfPropTrace$ 
satisfying $\propTrace \not\models \ltlFml$,
there exists $i \in \N$, such that 
for any $j > i$ and 
for any $\propTrace' \in \InfPropTrace$, we have  
$\subseq{\propTrace}{0}{j} \cdot \propTrace' \not\models \ltlFml$.
For a safety LTL formula $\ltlFml$, the violation of $\ltlFml$ can be monitored by a \emph{finite} prefix $\subseq{\propTrace}{0}{i} \in \FinPropTrace$ of $\propTrace \in \InfPropTrace$.
In~\cite{DBLP:conf/cav/dAmorimR05}, the finite semantics $\sem{\ltlFml}$ of LTL $\ltlFml$ is defined by the set of finite prefixes potentially satisfying the property $\ltlFml$.
We note that this semantics is also utilized in the latest version of LTSMin.

\begin{mydefinition}
 [finite semantics of LTL~\cite{DBLP:conf/cav/dAmorimR05}]
 For an LTL formula $\ltlFml$, $\sem{\ltlFml} \subseteq \FinPropTrace$
 is the following set of \emph{finite} sequences of subsets of atomic propositions
\[
 \sem{\ltlFml} = \{\propTrace \in \FinPropTrace \mid \exists \propTrace' \in \InfPropTrace.\, \propTrace\cdot\propTrace' \models \ltlFml \}
\]
\end{mydefinition}

\begin{mydefinition}
 [Mealy machine]
 \label{def:mealy_machine}
 For the input and output alphabet $\Alphabet$ and $\OAlphabet$,
 a \emph{Mealy machine} is a tuple $\ltsWithInside$, where
 $\Loc$ is the finite set of locations,
 $\initLoc \in \Loc$ is the initial location, and
 $\Transition\colon (\Loc \times \Alphabet) \to (\OAlphabet \times \Loc) $ is the transition function.
\end{mydefinition}

For a Mealy machine $\ltsWithInside$ over $\Alphabet$ and $\OAlphabet$, the \emph{language}
$\Lg(\lts) \subseteq \bigl(\Alphabet \times \OAlphabet\bigr)^{\omega}$ is 
$\Lg(\lts) = \{\ioChar[0],\ioChar[1],\dots \mid \exists \loc_1,\loc_2,\dots, \forall i \in \N.\, \Transition(\loc_i,\inChar_i) = (\outChar_i,\loc_{i+1})\}$.
For an infinite signal $\signal = \ioChar[0],\ioChar[1],\dots \in \bigl(\Alphabet \times \OAlphabet\bigr)^{\omega}$, we let
$\project{\signal}{1} = \inChar_0,\inChar_1,\dots \in \Alphabet^{\omega}$ and
$\project{\signal}{2} = \outChar_0,\outChar_1,\dots \in \OAlphabet^{\omega}$.
For a Mealy machine $\lts$,
the \emph{input language}
$\LgIn(\lts) \subseteq \Alphabet^{\omega}$ and
the \emph{output language}
$\LgOut(\lts) \subseteq \OAlphabet^{\omega}$ are 
$\LgIn(\lts) = \{\project{\signal}{1} \mid \exists \signal \in \Lg(\lts)\}$ and 
$\LgOut(\lts) = \{\project{\signal}{2} \mid \exists \signal \in \Lg(\lts)\}$.
We employ a Mealy machine over $\Alphabet$ and $\powerset{\AP}$ to model a system.

\begin{mydefinition}
 [LTL model checking] 
 Let $\Alphabet$ be the input alphabet and let $\AP$ be the set of the atomic propositions.
 Given 
 an LTL formula $\ltlFml$ over $\AP$ and
 a Mealy machine $\lts$ over $\Alphabet$ and $\powerset{\AP}$, \emph{LTL model checking} decides if we have
 \begin{math}
  \forall \propTrace \in \LgOut(\lts).\, \propTrace \models \ltlFml
 \end{math}.
 Moreover, it answers $\signal \in \Lg(\lts)$ satisfying $\project{\signal}{1} \not\models \ltlFml$ if such $\signal$ exists.
 We denote $\forall \propTrace \in \LgOut(\lts). \propTrace \models \ltlFml$ by $\lts \models \ltlFml$.
\end{mydefinition}

In the rest of this paper, we only consider \emph{safety} LTL formulas~\cite{DBLP:journals/fmsd/KupfermanV01}.
For any safety LTL formula $\ltlFml$, if we have $\lts \not\models \ltlFml$, there is a finite counterexample $\signal \in \bigl(\Alphabet \times \powerset{\AP}\bigr)^*$ such that 
$\project{\signal}{2} \in \FinPropTrace \setminus \sem{\ltlFml}$ and
there exists $\signal' \in \bigl(\Alphabet \times \powerset{\AP}\bigr)^{\omega}$ satisfying $\signal \cdot \signal' \in \Lg(\lts)$.
Thus, we use such a \emph{finite} counterexample $\signal\LongVersion{ \in \bigl(\Alphabet \times \powerset{\AP}\bigr)^*}$ as a witness of $\lts \not\models \ltlFml$.
We let $\LgFin(\lts) = \bigl\{\signal \in \bigl(\Alphabet \times \powerset{\AP}\bigr)^{*} \mid \exists \signal' \in \bigl(\Alphabet \times \powerset{\AP}\bigr)^{\omega}.\, \signal \cdot \signal' \in \Lg(\lts) \bigr\}$.

\subsection{Active automata learning}
\label{subsec:automata_learning}
\emph{Active automata learning} is an automata learning method pioneered by L* algorithm~\cite{DBLP:journals/iandc/Angluin87}, which learns the minimal DFA $\A_{\Lg}$ over $\Alphabet$ recognizing the target language $\Lg \subseteq \Alphabet^*$.
L* algorithm learns a DFA through the queries to \emph{membership} and \emph{equivalence} oracles.
Given a word $\word \in \Alphabet^*$, the membership oracle answers if\LongVersion{ $\word$ belongs to the target language $\Lg$ \ie{}} $\word \in \Lg$.
Given a candidate DFA $\A$, the equivalence oracle answers if\LongVersion{ $\A$ recognizes the target language $\Lg$ \ie{}} $\Lg(\A) = \Lg$,
where $\Lg(\A)$ is the language of the candidate DFA $\A$.
When $\A$ does not recognize $\Lg$,
the equivalence oracle also answers a counterexample $\word \in \Alphabet^*$ such that $\word \in \Lg(\A) \setdiff \Lg$.
We note that a Mealy machine $\lts$ can also be learned similarly\LongVersion{. See \eg}~\cite{DBLP:conf/sfm/SteffenHM11}.

\paragraph{Equivalence testing}
In practice, the target language $\Lg$ is usually given as a black-box system, and
a sound and complete equivalence oracle is often \emph{unimplementable} while
the given black-box system itself can be a membership oracle.
Therefore, we need an \emph{approximate} strategy for equivalence \emph{testing}.
For example, LearnLib~\cite{DBLP:conf/cav/IsbernerHS15} implements
\emph{deterministic exploration} (\eg{} complete, depth-bounded exploration),
\emph{random exploration} (\eg{} random words testing), and
\emph{conformance tests} (\eg{} W-method~\cite{DBLP:journals/tse/Chow78} and Wp-method~\cite{DBLP:journals/tse/FujiwaraBKAG91}).

\paragraph{Alphabet abstraction}
Another practical issue is that the input and output alphabet can be \emph{huge} or even \emph{infinite}, and the automata learning algorithm does not perform effectively or does not terminate.
For instance, the input and output of a CPS model is usually real-valued signals, which are infinitely many.
To overcome this issue, \emph{alphabet abstraction} is employed to reduce the alphabet size.
For example, a variant of Mealy machines is used to map the concrete and large alphabet to the abstract and small alphabet in~\cite{DBLP:journals/fmsd/AartsJUV15}.

\subsection{Black-box checking}

\emph{Black-box checking (BBC)}~\cite{DBLP:journals/jalc/PeledVY02}, is a black-box testing method\LongVersion{\footnote{Under some assumption, BBC is sound \ie{} BBC proves the correctness of the black-box system. See \eg~\cite{DBLP:journals/isse/MeijerP19}. However, the soundness assumption does not hold in most of the CPS application, and we use BBC just as a testing method. See also \cref{section:related_works}.}}
combining model checking and active automata learning.
Given a black-box and potentially infinite locations Mealy machine $\actualModel$ over $\Alphabet$ and $\powerset{\AP}$, and a safety LTL formula $\ltlFml$, 
BBC deems $\actualModel \models \ltlFml$ or returns a counterexample
$\signal \in \bigl(\Alphabet \times \powerset{\AP}\bigr)^*$ such that we have 
$\project{\signal}{2} \in \FinPropTrace \setminus \sem{\ltlFml}$ and
there exists $\signal' \in \bigl(\Alphabet \times \powerset{\AP}\bigr)^{\omega}$ satisfying $\signal \cdot \signal' \in \Lg(\lts)$.
In contrast to the usual testing methods, BBC also constructs a Mealy machine $\learnedModel$ through automata learning.
Thus, we can reuse some part of the previous testing results through the extracted Mealy machine $\learnedModel$.

\cref{fig:bbc} shows a workflow of BBC.
First, we learn a Mealy machine $\learnedModel$ from the black-box system $\actualModel$ by an automata learning algorithm \eg{} L*~\cite{DBLP:journals/iandc/Angluin87} or TTT algorithm~\cite{DBLP:conf/rv/IsbernerHS14}.
We note that the learned Mealy machine $\learnedModel$ may behave differently from the original\LongVersion{ black-box} system $\actualModel$ because our equivalence testing is an approximation, or even the equivalence testing might be omitted at this point.
Then, we check if we have $\learnedModel \models \ltlFml$ by model checking.
If we have $\learnedModel \not\models \ltlFml$, we also obtain a counterexample $\signal \in \bigl(\Alphabet \times \powerset{\AP}\bigr)^{*}$.
We feed the counterexample $\signal$ to the original system $\actualModel$ and check if $\signal$ is a witness of $\actualModel \not\models \ltlFml$, too.
If $\signal$ is a witness of $\actualModel \not\models \ltlFml$, we conclude $\actualModel \not\models \ltlFml$ and return the counterexample $\signal$.
Otherwise, we have $\signal \in \LgFin(\learnedModel)$ but $\signal \not\in \LgFin(\actualModel)$, and we use $\signal$ to refine our learning of $\learnedModel$.
If we have $\learnedModel \models \ltlFml$, we check if the behavior of $\actualModel$ and $\learnedModel$ are similar enough by equivalence testing.
If we find a counterexample $\signal \in \LgFin(\actualModel) \setdiff \LgFin(\learnedModel)$, we conclude that the learned Mealy machine $\learnedModel$ is not similar enough to the original system $\actualModel$, and
we use $\signal$ to refine our learning of $\learnedModel$.
If we could not find such $\signal$, we deem $\learnedModel$ to be equivalent to $\actualModel$  and return $\actualModel \models \ltlFml$, which is not always correct.

\section{Discrete-time signal temporal logic and robustness}
\label{section:discrete_stl}
\emph{Signal temporal logic (STL)}~\cite{DBLP:conf/formats/MalerN04} is a formalism to represent behavior of \emph{continuous-time}, real-valued signals
with quantitative satisfaction degree called \emph{robust semantics}~\cite{DBLP:conf/formats/DonzeM10}.
Due to the discrete nature of BBC, we need to represent \emph{discrete-time}, real-valued signals.
In this section, we introduce \emph{discrete-time} STL, which is a variant of LTL for real-valued signals.\LongVersion{ We define the robust semantics for both infinite and finite signals.}

\begin{mydefinition}
 [signal] 
 For a finite set $\OVar$ of variables, a \emph{signal} $\signal \in \OSignal$ is a (finite or infinite) sequence of valuations $\dval_i\colon \OVar \to \R$.
 For a finite signal $\signal \in \FinOSignal$,
 we denote the length $n$ of $\signalWithInside$ by $\abs{\signal}$.
\end{mydefinition}

\begin{mydefinition}
 [signal temporal logic]
 For a finite set $\OVar$ of variables, the syntax of \emph{signal temporal logic (STL)} is defined as follows, where $\ovar \in \OVar$, ${\bowtie} \in \{>,<\}$, $\constant \in \R$, and $i,j \in \N \cup \{+\infty\}$ satisfying $i < j$.
\[
 \stlFml, \stlFml' ::= \top \mid \ovar \bowtie \constant \mid \neg \stlFml \mid \stlFml \lor \stlFml' \mid \stlFml \UntilOp{\interval} \stlFml' \mid \NextOp \stlFml
\]
We use the following standard notation:
 $\bot \equiv \neg \top$;
 $y \geq c \equiv \neg (y < c)$;
 $y \leq c \equiv \neg (y > c)$;
 $\stlFml \land \stlFml' \equiv \neg ((\neg \stlFml) \lor (\neg\stlFml'))$;
 $\stlFml \Rightarrow \stlFml' \equiv (\neg \stlFml) \lor \stlFml'$;
 $\top \UntilOp{} \stlFml \equiv \top \UntilOp{\interval[0,\infty]} \stlFml$;
 $\DiaOp{\interval} \stlFml \equiv \top \UntilOp{\interval} \stlFml$; and
 $\BoxOp{\interval} \stlFml \equiv \neg (\DiaOp{\interval} \neg \stlFml)$.
\end{mydefinition}

For an STL formula $\stlFml$ over $\OVar$, an \emph{infinite} signal $\signalWithInfInside \in \InfOSignal$ over $\OVar$, and
 $k \in \N$, the satisfaction relation $\satisfy{\signal}{k}{\stlFml}$ is inductively defined as follows.
\begin{align*}
 \satisfy{\signal}{k}{\top} \qquad&
 \satisfy{\signal}{k}{\ovar > \constant} \iff \dval_k(\ovar) > \constant\\
 \satisfy{\signal}{k}{\ovar < \constant} \iff& \dval_k(\ovar) < \constant\\
 \satisfy{\signal}{k}{\neg \stlFml} \iff& \notsatisfy{\signal}{k}{\stlFml}\\
 \satisfy{\signal}{k}{\stlFml \lor \stlFml'} \iff& \satisfy{\signal}{k}{\stlFml} \lor \satisfy{\signal}{k}{\stlFml'}\\
 \satisfy{\signal}{k}{\NextOp{\stlFml}} \iff& \satisfy{\signal}{k+1}{\stlFml}\\
 \satisfy{\signal}{k}{\stlFml \UntilOp{\interval} \stlFml'} \iff& 
 \exists l \in \interval[k+i,k+j].\,\satisfy{\signal}{l}{\stlFml'} \\
 \land& \forall m \in \{k,k+1,\dots,l\}.\, \satisfy{\signal}{m}{\stlFml}
\end{align*}

The satisfaction relation $\satisfy{\signal}{k}{\stlFml}$ gives a \emph{qualitative} verdict of the satisfaction of the STL formula $\stlFml$ by the signal $\signal$.
The \emph{robust semantics} $\robust{\stlFml}{\signal}{k}$ gives a \emph{quantitative} satisfaction degree of the STL formula $\stlFml$ by the signal $\signal$.

\begin{mydefinition}
 [robust semantics]
 For an STL formula $\stlFml$ over $\OVar$, an \emph{infinite} signal $\signalWithInfInside \in \InfOSignal$ over $\OVar$, and $k \in \N$, the \emph{robust semantics} $\robust{\stlFml}{\signal}{k} \in \R \cup \{\pm\infty\}$ of the STL formula $\stlFml$ and the signal $\signal$ at $k$ is defined as follows.
 \begin{align*}
 \robust{\top}{\signal}{k} =& +\infty\qquad
 \robust{\NextOp \stlFml}{\signal}{k} =  \robust{\stlFml}{\signal}{k+1}\\
 \robust{\ovar > \constant}{\signal}{k} =&  \dval_k(\ovar) - \constant\quad
 \robust{\ovar < \constant}{\signal}{k}  = -\dval_k(\ovar) + \constant\\
 \robust{\neg \stlFml}{\signal}{k} =& -\robust{\stlFml}{\signal}{k}\\
 \robust{\stlFml \lor \stlFml'}{\signal}{k} =& \max(\robust{\stlFml}{\signal}{k}, \robust{\stlFml'}{\signal}{k})\\
 \robust{\stlFml \UntilOp{\interval} \stlFml'}{\signal}{k} =&\\
  \sup_{l \in \interval[k + i,k + j]}& \min\bigl(\robust{\stlFml'}{\signal}{l}, \min_{m \in \{k,k+1,\dots,l\}}\robust{\stlFml}{\signal}{m} \bigr)
\end{align*}
\end{mydefinition}

\begin{mytheorem}
 [soundness and completeness]
 \label{theorem:sound_correct_robust_infinite}
 For an STL formula $\stlFml$ over $\OVar$, an \emph{infinite} signal $\signalWithInfInside \in \InfOSignal$ over $\OVar$, and $k \in \N$ we have the following. 
 \begin{displaymath}
  \robust{\stlFml}{\signal}{k} > 0 \Rightarrow \satisfy{\signal}{k}{\stlFml} \quad
  \satisfy{\signal}{k}{\stlFml}  \Rightarrow \robust{\stlFml}{\signal}{k} \geq 0
 \end{displaymath}\qed
\end{mytheorem}

If we have $\satisfy{\signal}{0}{\stlFml}$, we denote $\signal \models \stlFml$.
The \emph{safety} fragment of STL is defined similarly to that of LTL.
For an STL formula $\stlFml$, 
we define two finite semantics: 
the \emph{supremum} finite semantics $\supSem{\stlFml} \subseteq \FinOSignal$ and
the \emph{infimum} finite semantics $\infSem{\stlFml} \subseteq \FinOSignal$.
The supremum finite semantics $\supSem{\stlFml}$ is the set of prefixes \emph{potentially} satisfying the property $\stlFml$, 
and corresponding to the finite semantics of LTL in~\cite{DBLP:conf/cav/dAmorimR05}.
The infimum finite semantics $\infSem{\stlFml}$ is the set of prefixes \emph{surely} satisfying the property $\stlFml$.

\begin{mydefinition}
 [$\supSem{\stlFml}, \infSem{\stlFml}$]
 For an STL formula $\stlFml$, 
 the \emph{supremum finite semantics} $\supSem{\stlFml} \subseteq \FinOSignal$ and
 the \emph{infimum finite semantics} $\infSem{\stlFml} \subseteq \FinOSignal$ are 
\ShortVersion{ $\supSem{\stlFml} = \{\signal \in \FinOSignal \mid \exists \signal' \in \InfOSignal.\, \signal \cdot\signal' \models \stlFml \}$ and
 $\infSem{\stlFml} = \{\signal \in \FinOSignal \mid \forall \signal' \in \InfOSignal.\, \signal \cdot\signal' \models \stlFml \}$.}\LongVersion{ defined as follows.
 \begin{align*}
  \supSem{\stlFml} &= \{\signal \in \FinOSignal \mid \exists \signal' \in \InfOSignal.\, \signal \cdot\signal' \models \stlFml \}\\
  \infSem{\stlFml} &= \{\signal \in \FinOSignal \mid \forall \signal' \in \InfOSignal.\, \signal \cdot\signal' \models \stlFml \}
 \end{align*}}
\end{mydefinition}

As the robust semantics for the finite signals,
we employ \emph{robust satisfaction interval (RoSI)}~\cite{DBLP:journals/fmsd/DeshmukhDGJJS17}.

\begin{mydefinition}
 [robust satisfaction interval]
 For an STL formula $\stlFml$ over $\OVar$, a \emph{finite} signal $\signal \in \FinOSignal$\LongVersion{ over $\OVar$}, and $k \in \N$,
 the \emph{robust satisfaction interval} $\RoSI{\stlFml}{\signal}{k}$ 
 is the following closed interval over $\R \cup \{\pm\infty\}$.
 \begin{displaymath}
  \RoSI{\stlFml}{\signal}{k} = \left[
  \inf_{\signal' \in \InfOSignal} \robust{\stlFml}{\signal \cdot \signal'}{k},
  \sup_{\signal' \in \InfOSignal} \robust{\stlFml}{\signal \cdot \signal'}{k}
  \right]
 \end{displaymath}
\end{mydefinition}

\begin{mytheorem}
 [soundness and completeness]
 \label{theorem:sound_correct_robust_finite}
 For
 an STL formula $\stlFml$ over $\OVar$, 
 a \emph{finite} signal $\signalWithInside \in \FinOSignal$\LongVersion{ over $\OVar$}, and
 $k \in \N$,
 we have the following. 
 \begin{align*}
  \sup(\RoSI{\stlFml}{\signal}{k}) > 0 &\Rightarrow \subseq{\signal}{k}{\abs{\signal}-1} \in \supSem{\stlFml}\\
  \subseq{\signal}{k}{\abs{\signal}-1} \in \supSem{\stlFml} &\Rightarrow \sup(\RoSI{\stlFml}{\signal}{k}) \geq 0 \\
  \inf(\RoSI{\stlFml}{\signal}{k}) > 0 &\Rightarrow \subseq{\signal}{k}{\abs{\signal}-1} \in \infSem{\stlFml}\\
  \subseq{\signal}{k}{\abs{\signal}-1} \in \infSem{\stlFml} &\Rightarrow \inf(\RoSI{\stlFml}{\signal}{k}) \geq 0
 \end{align*}
 \qed
\end{mytheorem}

One computational issue on\LongVersion{ the robust satisfaction interval} $\RoSI{\stlFml}{\signal}{k}$ is that its definition is not inductive and
it is unclear if it is effectively computable.
Instead, we use the following inductive overapproximation $\finRobust{\stlFml}{\signal}{k}$ of $\RoSI{\stlFml}{\signal}{k}$ as a quantitative satisfaction degree in our method.

\begin{mydefinition}
 [$\finRobust{\signal}{\stlFml}{k}$]
 For an STL formula $\stlFml$ over $\OVar$, 
 a \emph{finite} signal $\signalWithInside \in \FinOSignal$\LongVersion{ over $\OVar$}, and
 $k \in \N$, 
 $\finRobust{\stlFml}{\signal}{k}$ is 
 the closed interval over $\R \cup \{\pm\infty\}$ inductively defined as follows.
 \begin{align*}
 \finRobust{\top}{\signal}{k} &= [+\infty, +\infty]\\
 \finRobust{\ovar > \constant}{\signal}{k} &= 
 \begin{cases}
 [\dval_k(\ovar) - \constant, \dval_k(\ovar) - \constant] & \text{if $k < \abs{\signal}$}\\
 [-\infty, +\infty] & \text{if $k \geq \abs{\signal}$}
 \end{cases}\\
 \finRobust{\ovar < \constant}{\signal}{k} &= 
 \begin{cases}
 [- \dval_k(\ovar) + \constant, - \dval_k(\ovar) + \constant] & \text{if $k < \abs{\signal}$}\\
 [-\infty,+\infty] & \text{if $k \geq \abs{\signal}$}
 \end{cases}\\
 \finRobust{\neg \stlFml}{\signal}{k} &= -\finRobust{\stlFml}{\signal}{k}\\
 \finRobust{\stlFml \lor \stlFml'}{\signal}{k} &= \max(\finRobust{\stlFml}{\signal}{k}, \finRobust{\stlFml'}{\signal}{k})\\
 \finRobust{\NextOp \stlFml}{\signal}{k} &=  \finRobust{\stlFml}{\signal}{k+1}\\
 \finRobust{\stlFml \UntilOp{\interval} \stlFml'}{\signal}{k} &=\\ \max_{l \in \{k + i,k + i+1,\dots,k + j\}} &\min\bigl(\finRobust{\stlFml'}{\signal}{l}, \min_{m \in \{k,k+1,\dots,l\}}\finRobust{\stlFml}{\signal}{m} \bigr)
\end{align*}
\end{mydefinition}

\begin{mytheorem}
 \label{theorem:robust_finite_computation} 
 For any STL formula $\stlFml$ over $\OVar$, 
 a \emph{finite} signal $\signalWithInside \in \FinOSignal$\LongVersion{ over $\OVar$}, and
 $k \in \N$, 
 we have
 $\RoSI{\stlFml}{\signal}{k} \subseteq \finRobust{\stlFml}{\signal}{k}$.
 \qed
\end{mytheorem}

\mw{In a journal version, we can show an example of  $\RoSI{\stlFml}{\signal}{k} \neq \finRobust{\stlFml}{\signal}{k}$.}

\cref{corollary:robust_correctness} justifies the use of $\finRobust{\stlFml}{\signal}{0}$ as a quantitative satisfaction degree of $\signal \models \stlFml$.

\begin{mycorollary}
 \label{corollary:robust_correctness}
 Let $\stlFml$ be an STL formula over $\OVar$,
 let $\signalWithInside \in \FinOSignal$ be a finite signal\LongVersion{ over $\OVar$}, and
 $k \in \N$.
 If we have $\sup(\finRobust{\stlFml}{\signal}{0}) < 0$, 
 for any $\signal' \in \InfOSignal$, we have
 $\signal \cdot \signal' \not\models \stlFml$.
 If we have $\inf(\finRobust{\stlFml}{\signal}{0}) < 0$, 
 there exists $\signal' \in \InfOSignal$ satisfying
 $\signal \cdot \signal' \not\models \stlFml$.
 \qed
\end{mycorollary}

\section{Black-box checking of cyber-physical systems} \label{section:bbc-cps}

In this section, we show how to solve the falsification problem by BBC. 
Moreover, we enhance the membership testing by the robustness in STL.
This is our main contribution.
Let $\IVar$ and $\OVar$ be the finite sets of the input and output variables, respectively.
We define \emph{CPS model} $\actualModel$ over $(\IVar,\OVar)$ as a function $\actualModel\colon\FinISignal\to\FinOSignal$ satisfying $\abs{\signal} = \abs{\actualModel(\signal)}$.
The input signal $\signal \in \FinISignal$ shows the inputs (\eg{} the angle of the brake pedal) at each time step, and
 the output signal  $\actualModel(\signal) \in \FinOSignal$ shows the states (\eg{} the speed of the car) at each time step.

We solve the falsification problem using BBC, where the given black-box system is a CPS model $\actualModel\colon \FinISignal \to \FinOSignal$.
Although the input and the output domain of the CPS model $\actualModel$ is \emph{continuous}, 
we construct a Mealy machine $\learnedModel$ with \emph{finite} input and output.
In what follows, we present how to implement the membership and equivalence oracles for CPSs, and how we employ BBC to falsify multiple STL formulas.
We note that we can still use LTL model checking because discrete-time STL can be interpreted as LTL and 
we assume that there is a constant sampling rate for the CPS trajectories.

\subsection{Membership oracle with alphabet abstraction}
\begin{figure}[tbp]
\scalebox{0.85}{ 
\begin{tikzpicture}[shorten >=1pt,node distance=4cm,on grid,auto] 
 \small
  \node(outside){};
  \node[rectangle,draw,node distance=3.7cm](mapper)[right=of outside, align=center] {I/O Mapper\\ $\IMapper\colon \Alphabet \to \R^{\IVar}$ \\ $\OMapper\colon \R^{\OVar} \to \powerset{\AP}$};
  \node[rectangle,draw,node distance=4.0cm](model)[right=of mapper, align=center] {CPS Model $\actualModel$\\ \eg{} Simulink model};

  \node(title) at (2.0, 1.2) {\textbf{Membership Oracle}};
  \draw [rectangle] (0.3,1.4) rectangle (9.5,-1.2);

\path[->] 
  (outside) edge[transform canvas={yshift=3mm}] node[align=left] {$\icommand \in \ICommand$} (mapper)
  (mapper) edge[transform canvas={yshift=3mm}] node[above=0.3,align=left] {$\signal = \overline{\IMapper}(\icommand) \in \FinISignal$} (model)
  (model) edge[transform canvas={yshift=-3mm}] node[below=0.3,align=left] {$\actualModel(\signal) \in \FinOSignal$} (mapper)
  (mapper) edge[transform canvas={yshift=-3mm}] node[below=0.3,align=left,pos=0.2] {$\overline{\OMapper}(\actualModel(\signal)) \in \FinPropTrace$} (outside);
\end{tikzpicture}}
\caption{Membership oracle for a CPS model using alphabet abstraction: $\IMapper$ and $\OMapper$ are applied to each element. See \cref{section:preliminaries} for the notation $\overline{\IMapper}$ and $\overline{\OMapper}$.}
\label{fig:membership_mapper}
\end{figure}
Same as the usual BBC and automata learning for software testing,  we use the CPS model $\actualModel$ as the membership oracle.
As we discussed in \cref{subsec:automata_learning}, 
we have to abstract the alphabet due to the real-valued input and output of $\actualModel$.
As the abstract input and output alphabets, we use a finite set $\Alphabet$ and the power set $\powerset{\AP}$ of atomic propositions, respectively.
For simplicity, we employ a stateless mapper.
Namely, for the input alphabet $\Alphabet$, we define the \emph{input mapper} $\IMapper\colon \Alphabet \to \R^{\IVar}$, which assigns one input signal valuation to each $\inChar \in \Alphabet$, and
for the output alphabet $\powerset{\AP}$, we define the \emph{output mapper} $\OMapper\colon \R^{\OVar} \to \powerset{\AP}$, which returns the set of the atomic propositions satisfied for the given output signal valuation. 
We apply $\IMapper$ and $\OMapper$ to each element of the sequences.
See \cref{fig:membership_mapper} for an illustration.
We note that the construction of $\IMapper$ and $\OMapper$ as well as the choice of\LongVersion{ the input alphabet} $\Alphabet$ are done by a user.

\subsection{Robustness-guided equivalence testing} \label{section:rob-bbc}
\begin{algorithm}[tbp]
 \caption{Search-based equivalence testing}
 \label{algorithm:membership_testing}
\small
 \DontPrintSemicolon
 \newcommand{\myCommentFont}[1]{\texttt{\footnotesize{#1}}}
 \SetCommentSty{myCommentFont}
 \KwIn{CPS model $\actualModel\colon\R^{\IVar}\to\R^{\OVar}$, input mapper $\IMapper\colon \Alphabet \to \R^{\IVar}$, output mapper $\OMapper\colon \R^{\OVar} \to \powerset{\AP}$, STL formula $\stlFml$, and Mealy machine $\learnedModel\colon \ICommand\to\FinPropTrace$}
 \KwOut{Returns $\icommand \in \ICommand$ satisfying $\overline{\OMapper}(\actualModel(\overline{\IMapper}(\icommand))) \neq \learnedModel(\icommand)$, or $\bot$ when no such $\icommand$ was found}
 \tcc{sample the initial population}
 $\population \gets \initPopulation()$
 \label{alg_line:sample_initial_population}\;
 \Until{$\isTimeout()$} {
 \If{$\exists \icommand \in \population.\, \overline{\OMapper}(\actualModel(\overline{\IMapper}(\icommand))) \neq \learnedModel(\icommand)$} {
 \label{alg_line:check_equivalence}
 \Return{$\icommand$}
 }
 \tcc{Generate the next population \eg{} by random sampling or robustness-guided optimization}
 $\population \gets \generateNextPopulation(\population, \actualModel, \stlFml)$
 \label{alg_line:generate_next_population}
 }
 \Return{$\bot$}
\end{algorithm}

As we discussed in \cref{section:introduction,subsec:automata_learning}, 
we need an equivalence testing method to find a counterexample even if it is too rare for random search.
\cref{algorithm:membership_testing} shows a general outline of search-based equivalence testing (including random search) of a CPS model $\actualModel$ and a Mealy machine $\learnedModel$.

In random search, after randomly sampling the initial inputs $\population \subseteq \ICommand$ (\cref{alg_line:sample_initial_population}), 
we test the equivalence of $\actualModel$ and $\learnedModel$ for each input $\icommand \in \population$ (\cref{alg_line:check_equivalence}).
If we find no counterexample, we again randomly sample the next inputs $\population \subseteq \ICommand$ (\cref{alg_line:generate_next_population}) and 
test the equivalence again.
We repeat such a sampling (\cref{alg_line:generate_next_population}) and testing (\cref{alg_line:check_equivalence}) until we find a counterexample $\icommand$ or we reach the timeout.

The main observation in robustness-guided equivalence testing is as follows.
If we have $\overline{\OMapper} \circ \actualModel \circ \overline{\IMapper} \not\models \stlFml$ and $\learnedModel \models \stlFml$,
by a discrete input $\icommand \in \ICommand$ witnessing $\overline{\OMapper} \circ \actualModel \circ \overline{\IMapper} \not\models \stlFml$,
we can also witness $\overline{\OMapper} \circ \actualModel \circ \overline{\IMapper} \neq \learnedModel$, where $\circ$ is the function composition.
Thus, by minimizing the robustness of the CPS model $\actualModel$, we can guide the search to the inputs witnessing the difference between 
$\overline{\OMapper} \circ \actualModel \circ \overline{\IMapper}$ and $\learnedModel$.
Specifically, in \cref{alg_line:generate_next_population} of \cref{algorithm:membership_testing}, we use optimization to sample such inputs $\population$ that makes the robustness of the CPS model $\actualModel$ low.

For example, we can use \emph{local search} \eg{} \emph{hill climbing} and \emph{genetic algorithm}~\cite{DBLP:books/daglib/0070933}, where
the objective is to minimize $\sup(\finRobust{\stlFml}{\signal}{0})$.
We can continue this optimization along different equivalence testing calls by taking over the inputs in \cref{alg_line:sample_initial_population} instead of generating randomly.


\subsection{BBC for multiple specifications} \label{subsection:multiple-spec}
\LongVersion{\begin{algorithm}[tbp]
 \caption{BBC for multiple specifications}
 \label{algorithm:outline_BBC}
 \small
 \DontPrintSemicolon
 \newcommand{\myCommentFont}[1]{\texttt{\footnotesize{#1}}}
 \SetCommentSty{myCommentFont}
 \KwIn{
 CPS model $\actualModel\colon\R^{\IVar}\to\R^{\OVar}$, 
 input mapper $\IMapper\colon \Alphabet \to \R^{\IVar}$,
 output mapper $\OMapper\colon \R^{\OVar} \to \powerset{\AP}$, and
 STL formulas $\stlFml_1, \stlFml_2, \dots, \stlFml_n$.}
 \KwOut{A set $\mathit{result}$ of pairs $(\icommand_i,\stlFml_i)$, where $i \in \{1,2,\dots,n\}$ and
 $\icommand_i \in \ICommand$ is a witness of $\actualModel \not\models \stlFml_i$.}
 $\mathit{result} \gets \emptyset$;\,
 $\unfalsified \gets \{1,2,\dots,n\}$\;
 \tcc{Extract a Mealy machine from the actual CPS model (above of \cref{fig:bbc})}
 $\learnedModel \gets \mathrm{learnMealy}(\actualModel)$\;
 \label{alg_line:outline_BBC:init}
 \Repeat {$\mathit{cex} \neq \bot$} {
 $\mathit{cex} \gets \bot$\;
 \For{$i \in \mathrm{unfalsified}$}{
 \tcc{Model checking (center of \cref{fig:bbc})}
 \If{$\learnedModel \not\models \stlFml_i$}{
 \label{alg_line:outline_BBC:model_check}
 $\icommand_i \gets$ the witness of $\learnedModel \not\models \stlFml_i$\;
 \tcc{Feed $\icommand_i$ to $\actualModel$ (right of \cref{fig:bbc})}
 \If{$\actualModel \not\models \stlFml_i$ is witnessed by $\icommand_i$}{
 \label{alg_line:outline_BBC:falsified}
 \KwPush $(\icommand_i,\stlFml_i)$ \textbf{to} $\mathit{result}$\;
 \textbf{remove} $i$ \KwFrom $\unfalsified$
 } \lElse {
 $\mathit{cex} \gets \icommand_i$;\,
 \KwBreak
 }
 }
 }
 \If{$\mathit{cex} = \bot$} {
 \For{$i \in \mathrm{unfalsified}$} {
 \tcc{Search-based equivalence testing in \cref{section:rob-bbc} (left of \cref{fig:bbc})}
 $\mathit{cex} \gets \mathrm{searchEquivTest}(\actualModel, \IMapper, \OMapper, \stlFml_i, \learnedModel)$ \;
 \label{alg_line:outline_BBC:equiv}
 \If{$\mathit{cex} \neq \bot$}
 \KwBreak
 }
 }
 \If{$\mathit{cex} \neq \bot$} {
 $\learnedModel \gets \mathrm{learnMealy}(\actualModel, \learnedModel, \mathit{cex})$\;
 \label{alg_line:outline_BBC:refine}
 }
 }
 \Return{$\mathit{result}$}
\end{algorithm}}

\LongVersion{\cref{algorithm:outline_BBC} shows how we employ BBC to falsify multiple STL formulas.}
\ShortVersion{In what follows, we explain how we employ BBC to falsify multiple STL formulas.}
\ShortVersion{At first}\LongVersion{In \cref{alg_line:outline_BBC:init}}, we extract a Mealy machine $\learnedModel$ from the CPS model $\actualModel$.
Then,\LongVersion{ in \cref{alg_line:outline_BBC:model_check},} for each STL formula $\stlFml_i$ which is not falsified yet, we check if $\learnedModel \not\models\stlFml_i$ holds by LTL model checking.
When $\learnedModel \not\models\stlFml_i$ holds, we obtain a witness $\icommand_i \in \ICommand$ of $\learnedModel \not\models\stlFml_i$.
\ShortVersion{Then}\LongVersion{In \cref{alg_line:outline_BBC:falsified}}, we check if $\icommand_i$ also witnesses $\actualModel \not\models\stlFml_i$ by checking if we have
$\overline{\OMapper} (\actualModel (\overline{\IMapper} (\icommand_i))) \not\models\stlFml_i$.
When $\icommand_i$ also witnesses $\actualModel \not\models\stlFml_i$, we \ShortVersion{report $\icommand_i$}\LongVersion{store $\icommand_i$ in $\mathit{result}$} as a witness of $\actualModel \not\models\stlFml_i$.
Otherwise, we have $\overline{\OMapper}(\actualModel(\overline{\IMapper}(\icommand_i))) \neq \learnedModel(\icommand_i)$, and we use $\icommand_i$ to refine the learned Mealy machine $\learnedModel$\LongVersion{ (in \cref{alg_line:outline_BBC:refine})}.
When $\learnedModel \models \stlFml_i$ holds,\LongVersion{ in \cref{alg_line:outline_BBC:equiv},} we use the search-based equivalence testing (\cref{algorithm:membership_testing}) to find $\mathit{cex} \in \ICommand$ satisfying 
$\overline{\OMapper}(\actualModel(\overline{\IMapper}(\mathit{cex}))) \neq \learnedModel(\mathit{cex})$.
When we find such $\mathit{cex}$, we use it to refine the learned Mealy machine $\learnedModel$\LongVersion{ (in \cref{alg_line:outline_BBC:refine})}.
Otherwise, we deem $\overline{\OMapper} \circ \actualModel \circ \overline{\IMapper} = \learnedModel$ and \ShortVersion{finish BBC}\LongVersion{return $\mathit{result}$ as the final result: the set of the falsified specifications $\stlFml_i$ with inputs $\icommand_i$ witnessing $\actualModel \not\models \stlFml_i$}.

\section{Experimental Evaluation} \label{section:experiments}

We implemented a prototypical tool \ourtool{} for robustness-guided BBC of CPSs in Java using LearnLib~\cite{DBLP:conf/cav/IsbernerHS15}, jMetal~\cite{DBLP:journals/aes/DurilloN11}, and LTSMin~\cite{DBLP:conf/tacas/KantLMPBD15}.
As the optimization method in the robustness-guided equivalence testing\LongVersion{ (\ie{} \cref{alg_line:generate_next_population} of \cref{algorithm:membership_testing})}, we employ 
a hill climbing (\hc{}) and the genetic algorithm~\cite{DBLP:books/daglib/0070933} (\ga{}).
In \hc{}, for each discrete input sequence $\icommand \in \ICommand$ in the current population set,
we generate ``children'' input sequences by a random mutation.
Then,
we construct the next population set by taking the children with the smallest robust semantics.
\hc{} is one of the simplest algorithm to exploit the robust semantics, but we may get stuck in local optima.
In \ga{}, we avoid local optima by using larger population size and
combining mutation, crossover, and selection.
\FinalVersion{Our implementation is in \url{https://github.com/MasWag/FalCAuN}.}

We conducted experiments to answer the following research questions.
\begin{description}
 \item[RQ1] Does BBC falsify as many specifications as one of the state-of-the-art falsification tools?
 \item[RQ2] For which equivalence testing, BBC performs the best?
 \item[RQ3] Does BBC falsify multiple specifications effectively?
\end{description}

\paragraph{Benchmarks}
\begin{table*}[tbp]
 \caption{List of the STL formulas sets in our benchmarks. $\stlFml_1$, $\stlFml_2$, $\stlFml_3$, $\stlFml_4$, and $\stlFml_5$ are taken from \cite{DBLP:journals/tcad/ZhangESAH18}. The other benchmarks are original. 
 The STL formulas in $\stlFml_{6,\mathrm{tiny}}$, $\stlFml_{6,\mathrm{small}}$, $\stlFml_{6,\mathrm{medium}}$, $\stlFml_{6,\mathrm{large}}$, $\stlFml_{6,\mathrm{huge}}$, and $\stlFml_{6,\mathrm{gigantic}}$ have the same structure. 
 These benchmarks are mainly used to compare the scalability with respect to the size of the benchmark.
 }
 \label{table:benchmark_stl}
 \centering
 \small
 \begin{tabular}[t]{c|c|c|c}
  & STL template& parameter valuations & size \\\hline
  $\stlFml_1$ & $\BoxOp{} (v < p)$& $p \in \{100,102.5,105,107.5,110,112.5,115,117.5,120\}$ & $9$\\ 
  $\stlFml_2$ & $\BoxOp{} (g = 3 \Rightarrow v > p)$& $p \in \{20,22.5,25,27.5,30\}$ & $5$\\ 
  $\stlFml_3$ & $\DiaOp{[p_1,p_2]} (v < p_3 \lor v > p_4)$& $(p_1,p_2) \in \{(5,20),(5,25),(15,30),(10,30)\}, (p_3,p_4) \in \{(50,60), (53,57)\}$ & $8$\\ 
  $\stlFml_4$ & $\BoxOp{[0,26]} (v < p_1) \lor \BoxOp{[28,28]} (v > p_2)$ & $p_1 \in \{90,100,110\}, p_2 \in \{55,65,75\}$ & $9$\\ 
  $\stlFml_5$ & $\BoxOp{} (\omega < p_1 \lor \NextOp (\omega > p_2))$& $p_1 \in \{4000,4700\}, p_2 \in \{600,1000,1500\}$ & $6$\\ 
  $\stlFml_{6,\mathrm{tiny}}$ & $\BoxOp{} (v < p_1 \Rightarrow \BoxOp{[0,p_2]} (v < p_3))$ & $p_1 \in \{30,40\}, p_2 = 8,p_3 = 80$ & $2$\\ 
  $\stlFml_{6,\mathrm{small}}$ & $\BoxOp{} (v < p_1 \Rightarrow \BoxOp{[0,p_2]} (v < p_3))$ & $p_1 \in \{30,40\}, p_2 = 8, p_3 \in \{70,80\}$ & $4$\\ 
  $\stlFml_{6,\mathrm{medium}}$ & $\BoxOp{} (v < p_1 \Rightarrow \BoxOp{[0,p_2]} (v < p_3))$ & $p_1 \in \{30,40\}, p_2 \in \{8,10\},p_3 \in \{70,80\}$ & $8$\\ 
  $\stlFml_{6,\mathrm{large}}$ & $\BoxOp{} (v < p_1 \Rightarrow \BoxOp{[0,p_2]} (v < p_3))$ & $p_1 \in \{30,40,50\}, p_2 \in \{8,10\},p_3 \in \{70,80\}$ & $12$\\ 
  $\stlFml_{6,\mathrm{huge}}$ & $\BoxOp{} (v < p_1 \Rightarrow \BoxOp{[0,p_2]} (v < p_3))$ & $p_1 \in \{30,40,50\}, p_2 \in \{8,10\},p_3 \in \{60,70,80\}$ & $18$\\ 
  $\stlFml_{6,\mathrm{gigantic}}$ & $\BoxOp{} (v < p_1 \Rightarrow \BoxOp{[0,p_2]} (v < p_3))$ & $p_1 \in \{30,40,50\}, p_2 \in \{6,8,10\},p_3 \in \{60,70,80,90\}$ & $36$\\ 
  $\stlFml_7$ & $\BoxOp{} (((g \neq  p_1) \land \NextOp (g = p_1)) \Rightarrow \BoxOp{[0,p_2]} (g = p_1))$ & $p_1 \in \{1,2,3,4\}, p_2 \in \{1,2,3\}$ & $12$\\ 
 \end{tabular}
\end{table*}
As the CPS model $\actualModel$, we used the Simulink model of an automatic transmission system~\cite{DBLP:conf/cpsweek/HoxhaAF14}, which is one of the  standard models in the literature on falsification.
Given a 2-dimensional signal of \emph{throttle} and \emph{brake}, the automatic transmission model $\actualModel$ returns a 3-dimensional signal of \emph{velocity} $v$, \emph{rotation} $\omega$, and \emph{gear} $g$.
The range of throttle and brake are $[0,100]$ and $[0,325]$, respectively.
The domains of\LongVersion{ velocity} $v$ and\LongVersion{ rotation} $\omega$ are reals, and the domain of\LongVersion{ gear} $g$ is $\{1,2,3,4\}$.
As the specifications, we used the sets of the STL formulas in \cref{table:benchmark_stl}.
Each benchmark consists of multiple and similar STL formulas.
For example, $\stlFml_1$ consists of $6$ STL formulas and all of them are instances of the \emph{parametric} STL formula $\BoxOp{} (v < p)$.
This setting reflects our motivating example illustrated in \cref{section:introduction}: we do not know the exact threshold in the specification and we want to test the CPS model over various specification instances.
The benchmarks $\stlFml_1$--$\stlFml_5$ are taken from~\cite{DBLP:journals/tcad/ZhangESAH18} and the benchmarks $\stlFml_6$ and $\stlFml_7$ are our original.
\paragraph{Experiment}
We compared the robustness-guided BBC methods \hc{} and \ga{} with a baseline BBC method \random{} and one of the state-of-the-art falsification tools \breach{}.

In \hc{}, for each discrete input sequence $\icommand$ in the current population, 
we generate 60 ``children''\LongVersion{ discrete input sequences} by random swap:
given a discrete input sequence $\icommandWithInside$
random swap returns $\inChar_1,\inChar_2,\dots,\inChar_{i-1},\inChar,\inChar_{i+1},\dots,\inChar_n$, where
$\inChar \in \Alphabet$ and $i\in\{1,2,\dots,n\}$ are randomly chosen.
Among the ``children'' discrete input sequences, 5 input sequences realizing the smallest robust semantics are chosen to the next population.

In \ga{}, we used uniform mutation, uniform crossover, and tournament selection.
The population size, mutation probability, and crossover probability\LongVersion{ in \ga{}} are 150, 0.01, and 0.5\LongVersion{, respectively}.

In \random{}, we used a random equivalence testing.

We used TTT algorithm~\cite{DBLP:conf/rv/IsbernerHS14} for active automata learning in BBC.
For the experiments on BBC, the timeout is 4 hours \emph{in total}.
In BBC, the input length is fixed to 30.
The abstract alphabet $\Alphabet$ is $\abs{\Alphabet} = 4$ such that the throttle is either $0$ or $100$ and the brake is either $0$ or $325$.
The atomic propositions $\AP$ is the coarsest partitions of the output space (\ie{} the valuations of $v$,$\omega$, and $g$) compatible with the inequalities in STL formulas in each benchmark.

We used \breach~\cite{DBLP:conf/cav/Donze10} version 1.5.2 as a baseline.
\breach{} provides several optimization algorithms including
\emph{covariance matrix adaptation evolution strategy (CMA-ES)}~\cite{DBLP:conf/cec/AugerH05}, 
\emph{global Nelder-Mead (GNM)}\cite{luersen2004globalized}, and 
\emph{simulated annealing (SA)}~\cite{kirkpatrick1983optimization}.
Among them, we only used CMA-ES because it is reported to outperform the other optimization methods in~\cite{DBLP:conf/cav/ZhangHA19}.
For the experiment on \breach, the timeout is 15 minutes \emph{for each} specification.
In \breach{} we generated piecewise constant signals with 30 control points.
We note that the signals generated by \breach{} take floating-point values while the discrete input sequences generated by \ourtool{} take 4 values.
Thus, the search space of \breach{} is larger but there can be specifications falsifiable only by \breach{}.

\begin{table*}[tbp]
 \centering
\small
 \caption{Summary of the experiment result.
 The numbers $N/T$ in each cell are the number $N$ of the falsified specifications and the time $T$ [min.] to falsify all the falsifiable specification. 
 For each experiment setting, the average and the standard deviation are shown.
 For each benchmark $\stlFml_{i}$, the best cell in terms of the following order is highlighted: $N/T$ is better than $N'/T'$ if and only if we have $N>N'$ or we have both $N=N'$ and $T < T'$.
 For each benchmark, the largest average of the number of the falsified properties is shown in bold blue font.}
 \newcommand{\tbcolor}{\cellcolor{green!25}}
\newcommand{\mostFalsified}[1]{\textbf{\color{blue} #1}}
\begin{tabular}[t]{c||c||c|c|c|c|c|c|c|c}
  & {\footnotesize \textsc{PureRandom}} &\multicolumn{2}{c|}{\random} & \multicolumn{2}{c|}{\hc} & \multicolumn{2}{c|}{\ga} & \multicolumn{2}{c}{\breach}\\
  & {\footnotesize aver. \# of spec.} &average & std. dev. & average & std. dev. & average & std. dev. & average & std. dev. \\ \hline\hline
$\stlFml_1$&5.70&8.80/11.10&0.60/2.73&8.90/28.53&0.30/49.81&\mostFalsified{9.00}/68.96&0.00/64.64&\tbcolor \mostFalsified{9.00}/12.05&\tbcolor 0.00/0.19\\ \hline
$\stlFml_2$&0.00&\mostFalsified{4.90}/75.99&0.30/45.98&4.80/82.56&0.40/61.28&\tbcolor\mostFalsified{4.90}/74.12&\tbcolor0.30/77.88&2.00/0.20&0.00/0.00\\ \hline
$\stlFml_3$&0.00&\tbcolor \mostFalsified{8.00}/9.34&\tbcolor 0.00/2.88&\mostFalsified{8.00}/12.68&0.00/4.45&\mostFalsified{8.00}/12.87&0.00/4.96&\mostFalsified{8.00}/22.43&0.00/0.58\\ \hline
$\stlFml_4$&0.60&6.10/100.83&0.70/76.80&5.90/124.88&0.70/73.43&\tbcolor \mostFalsified{6.90}/163.03&\tbcolor 0.30/24.56&2.60/22.37&0.80/7.41\\ \hline
$\stlFml_5$&2.40&\mostFalsified{6.00}/139.72&0.00/132.73&3.30/72.99&2.49/124.15&\tbcolor\mostFalsified{6.00}/133.66&\tbcolor 0.00/140.54&3.00/5.78&0.00/0.45\\ \hline
$\stlFml_{6,\mathrm{tiny}}$&\mostFalsified{2.00}&\tbcolor \mostFalsified{2.00}/2.24&\tbcolor 0.00/1.14&\mostFalsified{2.00}/2.44&0.00/1.11&\mostFalsified{2.00}/3.54&0.00/1.47&\mostFalsified{2.00}/3.12&0.00/0.09\\ \hline
$\stlFml_{6,\mathrm{small}}$&\mostFalsified{4.00}&\mostFalsified{4.00}/2.98&0.00/1.38&\tbcolor \mostFalsified{4.00}/2.58&\tbcolor 0.00/1.44&\mostFalsified{4.00}/3.20&0.00/1.03&\mostFalsified{4.00}/4.41&0.00/0.18\\ \hline
$\stlFml_{6,\mathrm{medium}}$&6.10&7.20/141.83&2.40/416.15&\tbcolor \mostFalsified{8.00}/2.50&\tbcolor 0.00/1.31&\mostFalsified{8.00}/4.07&0.00/2.52&\mostFalsified{8.00}/7.74&0.00/0.04\\ \hline
$\stlFml_{6,\mathrm{large}}$&9.00&10.80/288.46&3.60/566.25&\tbcolor \mostFalsified{12.00}/3.00&\tbcolor 0.00/2.02&\mostFalsified{12.00}/3.47&0.00/1.46&\mostFalsified{12.00}/9.99&0.00/0.04\\ \hline
$\stlFml_{6,\mathrm{huge}}$&12.00&\mostFalsified{18.00}/2.36&0.00/1.21&\tbcolor \mostFalsified{18.00}/2.00&\tbcolor 0.00/0.74&\mostFalsified{18.00}/3.21&0.00/0.78&\mostFalsified{18.00}/12.45&0.00/0.06\\ \hline
$\stlFml_{6,\mathrm{gigantic}}$&30.00&\tbcolor \mostFalsified{31.00}/5.59&\tbcolor 0.00/1.93&\mostFalsified{31.00}/12.15&0.00/10.05&\mostFalsified{31.00}/7.95&0.00/3.93&\mostFalsified{31.00}/36.30&0.00/0.45\\ \hline
$\stlFml_7$&0.00&\mostFalsified{12.00}/1.35&0.00/0.76&\tbcolor \mostFalsified{12.00}/1.25&\tbcolor 0.00/0.72&\mostFalsified{12.00}/1.84&0.00/0.50&9.00/0.38&0.00/0.01
\end{tabular}
 
 \label{table:experiment_result}
\end{table*}

Since the optimization algorithm in \ga{}, \hc{}, \breach{} as well as the random sampling in \random{} are stochastic, 
we executed each benchmark and algorithm for 10 times.
For each execution, we measured the number of the falsified specifications and the time to falsify all the falsified specifications.
For \breach{}, we used the sum of the time to falsify all the falsified specifications.
\cref{table:experiment_result} shows the summary of the experiment result.
We also show the result of a pure random sampling process (\textsc{PureRandom}) to confirm the hardness of the benchmarks.
We also note that $\varphi_1$ and $\varphi_7$ contain \textrm{AT1} and a variant of \textrm{AT5} specifications in~\cite{DBLP:conf/cpsweek/ErnstADFMPYYZ19}.
Both of the specifications are falsified by \ga{} 10 times out of 10 trials.
We conducted the experiments on an Amazon EC2 c4.large instance (2 vCPUs and 3.75 GiB RAM).

\subsection{RQ1: Comparison with \breach}

In \cref{table:experiment_result}, we observe that on average, \ga{} falsified as many properties as $\breach$ does for any benchmark $\stlFml_i$.
\hc{} also falsified as many properties as $\breach$ does for any benchmark $\stlFml_i$ except for $\stlFml_1$.
Even for $\stlFml_1$, the number of the falsified properties of $\hc$ is comparable to that of $\breach$.
We also observe that $\random$ falsified as many properties as $\breach$ except for $\stlFml_1$, $\stlFml_{6,\mathrm{medium}}$, and $\stlFml_{6,\mathrm{large}}$.

One reason of the good performance of \ga{} and \hc{} is that 
the equivalence testing in these methods utilizes a \emph{discrete} optimization and tends to work well
even if the different part of the input sequence contributes to the robust semantics differently.
For example,\LongVersion{ in order} to falsify $\stlFml_4$, we have to find an input that makes the velocity high in the beginning and suddenly decreases the velocity at $28$ time units.
Such an optimization is not easy for \emph{continuous} optimization methods \eg{} CMA-ES. 

Another reason is that CMA-ES does not work well when the fitness function has very small slope.
For example, for the benchmark $\stlFml_2$, when the gear is not 3, the change of the robustness is almost discrete and the slope can be 0.
This is a difficult situation for many continuous optimization methods based on the slope.
Especially when the slope is too small, CMA-ES stops deeming there is no better inputs.
On the other hand, the behavior of the robustness-guided equivalence checking methods is much like the random search and it successfully falsified the specifications.

\subsection{RQ2: Best equivalence testing method}

In \cref{table:experiment_result}, we observe that on average, the number of the falsified properties of \ga{} is greater than or equal to that of \random{} and \hc.
Moreover, \ga{} has smaller standard deviation of the number of the properties than \random{} and \hc.
This is because \ga{} has a good balance of exploitation of exploration and the equivalence testing tends have a good performance constantly while
\random{} and \hc{} occasionally fails to find a counterexample in the equivalence testing.

\begin{table*}[tbp]
 \centering
 \caption{Result of falsification only using the extracted Mealy machine.
 The second column shows the number of the counterexamples found by model checking of the Mealy machines $\learnedModel$ extracted during the BBC.
 The third column shows the number of the actual counterexamples confirmed through a simulation of the CPS model $\actualModel$.
 The fourth and the fifth columns show the average and the standard deviation of the robustness, respectively.}
 \label{table:falsify_by_model_checking}
\small
 \begin{tabular}{c|c|c|c|c}
  STL formula $\stlFml$ & \# of $\stlFml \not\models\learnedModel$ & \# of $\stlFml \not\models\actualModel$ & Average of $\llbracket \stlFml \rrbracket$& std. dev. of $\llbracket \stlFml \rrbracket$\\\hline
  $\BoxOp{} (v < 90)$  	& 10& 5 &1.10                                       &1.94 \\ 
  $\BoxOp{[0,26]} (v < 90) \lor \BoxOp{[28,28]} (v > 40)$	&4  &0 &4.19 &0.00 \\ 
  $\BoxOp{[0,26]} (v < 90) \lor \BoxOp{[28,28]} (v > 50)$	&10 &0 &3.80 &0.60 \\ 
  $\BoxOp{[0,26]} (v < 90) \lor \BoxOp{[28,28]} (v > 60)$	&10 &0 &3.24 &0.76    
 \end{tabular}
\end{table*}

On the other hand, we also observe that $\ga$ tends not to be the fastest among the BBC methods.
This makes the number of the highlighted cells of $\ga$ smaller than that of $\hc$ and equal to that of $\random$ although $\ga$ falsified the largest number of properties.
This is because the genetic algorithm in \ga{} is more complicated than the hill climbing in $\hc$ and the random search in $\random$ while these simple optimization is enough for easy benchmarks.
However, even though $\ga$ is not the fastest BBC method, the additional time caused by $\ga$ is only a few minutes and it is acceptable for many practical usages.
Therefore, we conclude that $\ga$ performed the best among the three BBC methods.

\subsection{RQ3: Effectiveness to falsify multiple specifications}

\begin{figure}[tbp]
 \centering
 \scalebox{0.58}{\begin{tikzpicture}[gnuplot]
\tikzset{every node/.append style={font={\fontsize{15.0pt}{18.0pt}\selectfont}}}
\path (0.000,0.000) rectangle (12.500,8.750);
\gpcolor{color=gp lt color border}
\gpsetlinetype{gp lt border}
\gpsetdashtype{gp dt solid}
\gpsetlinewidth{1.00}
\draw[gp path] (1.704,1.478)--(1.884,1.478);
\draw[gp path] (11.671,1.478)--(11.491,1.478);
\node[gp node right,font={\huge}] at (1.428,1.478) {$0$};
\draw[gp path] (1.704,2.271)--(1.884,2.271);
\draw[gp path] (11.671,2.271)--(11.491,2.271);
\node[gp node right,font={\huge}] at (1.428,2.271) {$5$};
\draw[gp path] (1.704,3.065)--(1.884,3.065);
\draw[gp path] (11.671,3.065)--(11.491,3.065);
\node[gp node right,font={\huge}] at (1.428,3.065) {$10$};
\draw[gp path] (1.704,3.858)--(1.884,3.858);
\draw[gp path] (11.671,3.858)--(11.491,3.858);
\node[gp node right,font={\huge}] at (1.428,3.858) {$15$};
\draw[gp path] (1.704,4.652)--(1.884,4.652);
\draw[gp path] (11.671,4.652)--(11.491,4.652);
\node[gp node right,font={\huge}] at (1.428,4.652) {$20$};
\draw[gp path] (1.704,5.445)--(1.884,5.445);
\draw[gp path] (11.671,5.445)--(11.491,5.445);
\node[gp node right,font={\huge}] at (1.428,5.445) {$25$};
\draw[gp path] (1.704,6.238)--(1.884,6.238);
\draw[gp path] (11.671,6.238)--(11.491,6.238);
\node[gp node right,font={\huge}] at (1.428,6.238) {$30$};
\draw[gp path] (1.704,7.032)--(1.884,7.032);
\draw[gp path] (11.671,7.032)--(11.491,7.032);
\node[gp node right,font={\huge}] at (1.428,7.032) {$35$};
\draw[gp path] (1.704,7.825)--(1.884,7.825);
\draw[gp path] (11.671,7.825)--(11.491,7.825);
\node[gp node right,font={\huge}] at (1.428,7.825) {$40$};
\draw[gp path] (1.704,1.478)--(1.704,1.658);
\draw[gp path] (1.704,7.825)--(1.704,7.645);
\node[gp node center,font={\huge}] at (1.704,1.016) {$0$};
\draw[gp path] (3.128,1.478)--(3.128,1.658);
\draw[gp path] (3.128,7.825)--(3.128,7.645);
\node[gp node center,font={\huge}] at (3.128,1.016) {$5$};
\draw[gp path] (4.552,1.478)--(4.552,1.658);
\draw[gp path] (4.552,7.825)--(4.552,7.645);
\node[gp node center,font={\huge}] at (4.552,1.016) {$10$};
\draw[gp path] (5.976,1.478)--(5.976,1.658);
\draw[gp path] (5.976,7.825)--(5.976,7.645);
\node[gp node center,font={\huge}] at (5.976,1.016) {$15$};
\draw[gp path] (7.399,1.478)--(7.399,1.658);
\draw[gp path] (7.399,7.825)--(7.399,7.645);
\node[gp node center,font={\huge}] at (7.399,1.016) {$20$};
\draw[gp path] (8.823,1.478)--(8.823,1.658);
\draw[gp path] (8.823,7.825)--(8.823,7.645);
\node[gp node center,font={\huge}] at (8.823,1.016) {$25$};
\draw[gp path] (10.247,1.478)--(10.247,1.658);
\draw[gp path] (10.247,7.825)--(10.247,7.645);
\node[gp node center,font={\huge}] at (10.247,1.016) {$30$};
\draw[gp path] (11.671,1.478)--(11.671,1.658);
\draw[gp path] (11.671,7.825)--(11.671,7.645);
\node[gp node center,font={\huge}] at (11.671,1.016) {$35$};
\draw[gp path] (1.704,7.825)--(1.704,1.478)--(11.671,1.478)--(11.671,7.825)--cycle;
\node[gp node center,rotate=-270,font={\huge}] at (0.369,4.651) {The execution time [min.]};
\node[gp node center,font={\huge}] at (6.687,0.323) {The number of the falsified properites};
\node[gp node left,font={\large}] at (2.367,8.339) {\random\hspace{0mm}};
\gpcolor{rgb color={0.580,0.000,0.827}}
\gpsetlinewidth{5.00}
\draw[gp path] (1.635,8.339)--(2.091,8.339);
\draw[gp path] (2.274,1.833)--(2.843,1.951)--(3.085,7.825);
\draw[gp path] (6.560,7.825)--(6.830,1.852)--(10.532,2.365);
\gpsetpointsize{8.00}
\gppoint{gp mark 5}{(2.274,1.833)}
\gppoint{gp mark 5}{(2.843,1.951)}
\gppoint{gp mark 5}{(6.830,1.852)}
\gppoint{gp mark 5}{(10.532,2.365)}
\gppoint{gp mark 5}{(1.863,8.339)}
\gpcolor{color=gp lt color border}
\node[gp node left,font={\large}] at (5.031,8.339) {\hc};
\gpcolor{rgb color={0.000,0.620,0.451}}
\draw[gp path] (4.299,8.339)--(4.755,8.339);
\draw[gp path] (2.274,1.865)--(2.843,1.887)--(3.982,1.875)--(5.121,1.954)--(6.830,1.795)%
  --(10.532,3.406);
\gppoint{gp mark 7}{(2.274,1.865)}
\gppoint{gp mark 7}{(2.843,1.887)}
\gppoint{gp mark 7}{(3.982,1.875)}
\gppoint{gp mark 7}{(5.121,1.954)}
\gppoint{gp mark 7}{(6.830,1.795)}
\gppoint{gp mark 7}{(10.532,3.406)}
\gppoint{gp mark 7}{(4.527,8.339)}
\gpcolor{color=gp lt color border}
\node[gp node left,font={\large}] at (7.695,8.339) {\ga};
\gpcolor{rgb color={0.337,0.706,0.914}}
\draw[gp path] (6.963,8.339)--(7.419,8.339);
\draw[gp path] (2.274,2.040)--(2.843,1.986)--(3.982,2.124)--(5.121,2.029)--(6.830,1.987)%
  --(10.532,2.739);
\gppoint{gp mark 6}{(2.274,2.040)}
\gppoint{gp mark 6}{(2.843,1.986)}
\gppoint{gp mark 6}{(3.982,2.124)}
\gppoint{gp mark 6}{(5.121,2.029)}
\gppoint{gp mark 6}{(6.830,1.987)}
\gppoint{gp mark 6}{(10.532,2.739)}
\gppoint{gp mark 6}{(7.191,8.339)}
\gpcolor{color=gp lt color border}
\node[gp node left,font={\large}] at (10.359,8.339) {\breach};
\gpcolor{rgb color={0.902,0.624,0.000}}
\draw[gp path] (9.627,8.339)--(10.083,8.339);
\draw[gp path] (2.274,1.973)--(2.843,2.178)--(3.982,2.706)--(5.121,3.063)--(6.830,3.454)%
  --(10.532,7.238);
\gppoint{gp mark 4}{(2.274,1.973)}
\gppoint{gp mark 4}{(2.843,2.178)}
\gppoint{gp mark 4}{(3.982,2.706)}
\gppoint{gp mark 4}{(5.121,3.063)}
\gppoint{gp mark 4}{(6.830,3.454)}
\gppoint{gp mark 4}{(10.532,7.238)}
\gppoint{gp mark 4}{(9.855,8.339)}
\gpcolor{color=gp lt color border}
\gpsetlinewidth{1.00}
\draw[gp path] (1.704,7.825)--(1.704,1.478)--(11.671,1.478)--(11.671,7.825)--cycle;
\gpdefrectangularnode{gp plot 1}{\pgfpoint{1.704cm}{1.478cm}}{\pgfpoint{11.671cm}{7.825cm}}
\end{tikzpicture}
 }
 \caption{The average of the number of the falsified properties and the time to falsify them [min.] for
 $\stlFml_{6,\mathrm{tiny}}$,
 $\stlFml_{6,\mathrm{small}}$,
 $\stlFml_{6,\mathrm{medium}}$, 
 $\stlFml_{6,\mathrm{large}}$,
 $\stlFml_{6,\mathrm{huge}}$, and
 $\stlFml_{6,\mathrm{gigantic}}$.}
 \label{figure:M3-result}
\end{figure}

\cref{figure:M3-result} shows the average of the number of the falsified properties and the time to falsify these properties for
 $\stlFml_{6,\mathrm{tiny}}$,
 $\stlFml_{6,\mathrm{small}}$,
 $\stlFml_{6,\mathrm{medium}}$, 
 $\stlFml_{6,\mathrm{large}}$,
 $\stlFml_{6,\mathrm{huge}}$, and
 $\stlFml_{6,\mathrm{gigantic}}$.
We observe that except for $\stlFml_{6,\mathrm{medium}}$ and $\stlFml_{6,\mathrm{large}}$ of $\hc$, the execution time of the BBC algorithms tends to be shorter than that of $\breach$.
Especially, for
 $\stlFml_{6,\mathrm{tiny}}$,
 $\stlFml_{6,\mathrm{small}}$,
 $\stlFml_{6,\mathrm{medium}}$, 
 $\stlFml_{6,\mathrm{large}}$, and
 $\stlFml_{6,\mathrm{huge}}$, we observe that the execution time of $\hc$ and $\ga$ is more or less constant while the execution time of $\breach$ increases linearly. 
 This is because in BBC, once we learn a sufficiently accurate Mealy machine $\learnedModel$, we often find counterexamples for several specifications immediately.
On the other hand, in $\breach$, each falsification trial is independent and the execution time increases linearly.
We note that the huge execution time of $\random$ for 
 $\stlFml_{6,\mathrm{small}}$ and $\stlFml_{6,\mathrm{medium}}$ is due to the outliers 
as the large standard deviations suggest.

\subsection{Discussion on the extracted Mealy machines}
\label{subsec:extracted_mealy_machine}
\begin{table}[tbp]
 \centering
 \caption{Average of the number of the states of the extracted Mealy machine}
 \footnotesize
 \begin{tabular}[t]{c||c|c|c}
 & \random & \hc & \ga\\ \hline\hline
 $\stlFml_1$&181.90&270.90&441.50\\ \hline
$\stlFml_2$&612.60&661.60&610.00\\ \hline
$\stlFml_3$&154.20&200.20&198.30\\ \hline
$\stlFml_4$&1372.70&1194.70&1353.30\\ \hline
$\stlFml_5$&948.60&1442.14&888.60\\ \hline
$\stlFml_{6,\mathrm{tiny}}$&26.60&32.10&35.20\\ \hline
$\stlFml_{6,\mathrm{small}}$&45.30&40.60&39.50\\ \hline
$\stlFml_{6,\mathrm{medium}}$&41.44&37.80&47.70\\ \hline
$\stlFml_{6,\mathrm{large}}$&32.89&48.70&44.80\\ \hline
$\stlFml_{6,\mathrm{huge}}$&41.20&36.70&44.80\\ \hline
$\stlFml_{6,\mathrm{gigantic}}$&1912.00&1714.40&1891.10\\ \hline
$\stlFml_7$&24.00&21.10&20.00
\end{tabular}

 \label{table:extracted_states}
\end{table}

One natural question on BBC is whether the extracted Mealy machine $\learnedModel$ is a good approximation of the original system $\actualModel$.
Especially, since the robustness-guided equivalence testing focuses on the inputs realizing \emph{low} robustness, 
it is unclear if the extracted Mealy machine $\learnedModel$ behaves similarly to the original system $\actualModel$ even for the inputs not realizing low robustness.
We note that as shown in \cref{table:extracted_states}, the extracted Mealy machines tend to be huge and a manual inspection is unrealistic.

In order to obtain insights on the aforementioned question, we conducted the following additional experiments.
\begin{ienumeration}
 \item For a Mealy machine $\learnedModel$ generated through BBC and an STL formula $\stlFml$ \emph{not} used when $\learnedModel$ is learned,
 we conducted model checking to obtain a witness $\icommand \in \ICommand$ of $\learnedModel \not\models \stlFml$. 
 We note that if we have $\learnedModel \models \stlFml$, we cannot obtain such $\icommand$.
 \item By feeding the generated witness $\icommand \in \ICommand$ to the original system $\actualModel$, we checked if $\icommand$ also witnesses $\actualModel \models \stlFml$.
 Precisely, we checked if we have $\overline{\OMapper} (\actualModel (\overline{\IMapper} (\icommand)) ) \models \stlFml$ by running a simulation.
\end{ienumeration}
As the Mealy machines, we used the 10 Mealy machines generated by $\ga$ with the benchmark $\stlFml_{6,\mathrm{gigantic}}$.
As the STL formulas, we used variants of the STL formulas in $\stlFml_1$ and $\stlFml_4$.

\cref{table:falsify_by_model_checking} shows the experiment result.
In the second column of \cref{table:falsify_by_model_checking}, we observe that we tend to be able to falsify the STL formula $\stlFml$ with respect to the extracted Mealy machine $\learnedModel$.
On the other hand, in the third column of \cref{table:falsify_by_model_checking}, we observe that the witness $\icommand \in \ICommand$ of $\stlFml \not \models\learnedModel$ is usually not a witness of $\stlFml \not \models\actualModel$.
This suggests that if we directly reuse a Mealy machine generated through BBC of different STL formulas, falsification does not perform well.
However, in the fourth column, we observe that the robustness is much smaller than the threshold in the STL formulas, and
 the witness $\icommand$ of $\stlFml \not \models\learnedModel$ actually witnesses ``near violation'' of $\stlFml \models\learnedModel$.
We note that this is not due to outliers as we observe the small standard deviation in the fifth column, 
Therefore, it seems that
the extracted Mealy machine $\learnedModel$ is not a very precise abstraction of the original system $\actualModel$, but we can potentially use $\learnedModel$ as a rough approximation of $\actualModel$.

\section{Related works} \label{section:related_works}

Black-box checking (BBC)~\cite{DBLP:journals/jalc/PeledVY02} (or learning-based testing (LBT)~\cite{DBLP:conf/dagstuhl/Meinke16}) is initially presented as a \emph{sound} black-box testing method utilizing
Vasilevskii and Chow (VC) algorithm~\cite{Vasilevskii1973,DBLP:journals/tse/Chow78} as the equivalence oracle.
The correctness of the VC algorithm relies on the upper bound of the size of the state space of the black-box system.
In~\cite{DBLP:journals/isse/MeijerP19}, B\"uchi acceptance condition in the state space of the black-box system is used for the sound equivalence checking.

A great effort has been devoted to a more practical direction of BBC, including the testing of automotive systems.
For example,  case studies on testing of automotive software systems are shown in~\cite{DBLP:conf/imbsa/KhosrowjerdiMR17} and
an application to the CPSs with continuous dynamics is presented in~\cite{DBLP:conf/epew/Meinke17,DBLP:conf/kbse/KhosrowjerdiM18}.
However, up to our knowledge, there is no work exploiting the quantitative satisfaction degree of the requirements in addition to Boolean satisfaction.
For BBC,
as far as we are aware of, two tools have been presented: LBTest~\cite{DBLP:conf/icst/MeinkeS13} and an implementation~\cite{DBLP:journals/isse/MeijerP19} in LearnLib~\cite{DBLP:conf/cav/IsbernerHS15}.
Our prototypical tool \ourtool{} relies on the implementation~\cite{DBLP:journals/isse/MeijerP19} in LearnLib.

Falsification is one of the well-known quality assurance methods of CPSs with two well-matured tools: \breach~\cite{DBLP:conf/cav/Donze10} and \staliro~\cite{DBLP:conf/tacas/AnnpureddyLFS11}.
Moreover, a friendly competition~\cite{DBLP:conf/cpsweek/2019arch} has been held every year since 2017.

Among many algorithms for falsification, only a few algorithms utilize model learning.
For example, in \cite{DBLP:journals/tecs/DeshmukhHJMP17}, for a CPS model $\actualModel$ and an STL formula $\stlFml$, a probabilistic model is constructed to approximate the function from an input signal $\signal$ to the robust semantics of $\stlFml$ over the output signal $\actualModel (\signal)$, and \emph{Bayesian optimization}~\cite{DBLP:journals/corr/abs-1012-2599} is used to make falsification efficient.
In~\cite{DBLP:conf/fm/AkazakiLYDH18},
\emph{deep reinforcement learning}~\cite{DBLP:journals/nature/MnihKSRVBGRFOPB15} is used for a similar optimization.
One drawback of these algorithms is that the learned model depends on the STL formula $\stlFml$, and it is (at least) not straightforward to apply for the falsification of multiple STL formulas.

In~\cite{DBLP:conf/cpsweek/KatoIH18}, reinforcement learning is used to falsify one specification for multiple but similar systems effectively.
We note that our BBC approach is also applicable for falsification of multiple but similar systems by \emph{adaptive model checking}~\cite{DBLP:journals/igpl/GrocePY06}.


\section{Conclusions and future work} \label{section:conclusions_and_future_work}

Combining optimization-based falsification and black-box checking (BBC), we presented robustness-guided BBC, which is a method to falsify multiple specifications efficiently.
Our main technical contribution is to use the robust semantics of STL to enhance the equivalence testing in active automata learning.
Our experiment results suggest that robustness-guided BBC by genetic algorithm ($\ga$) tends to outperform baseline algorithms of both\LongVersion{ optimization-based} falsification and BBC.
\LongVersion{Namely, we compared with $\breach$, which is one of the state-of-the-art falsification tools, and $\random$, which is a BBC method with random equivalence testing.}

One future direction is to reuse the extracted Mealy machine $\learnedModel$ for BBC over the STL formulas $\stlFml$ other than the formulas $\stlFml'$ examined when $\learnedModel$ is extracted.
As we observed in \cref{subsec:extracted_mealy_machine}, $\learnedModel$ may not be a good approximation of $\actualModel$ for falsification of $\stlFml$, but it seems $\learnedModel$ roughly captures the behavior of $\actualModel$.
Thus, we need to (hopefully only slightly) refine $\learnedModel$ to obtain a witness of $\actualModel \not\models \stlFml$.
When $\learnedModel \models \stlFml$ holds, we have to find an input to refine $\learnedModel$ by robustness-guided equivalence testing.
It is an interesting future work to make this robustness-guided equivalence testing efficient utilizing $\learnedModel$.
We note that when we have $\learnedModel \not\models \stlFml$, we can use the counterexample obtained by the model checking to $\learnedModel$.
It is also a future work to use $\learnedModel$ to explain why the BBC failed.


Another future direction is an efficient falsification method over a family of similar systems using \emph{adaptive model checking}~\cite{DBLP:journals/igpl/GrocePY06}.

It is also a future work to conduct further detailed experimental evaluation to compare with more tools by using the ARCH-COMP benchmark~\cite{DBLP:conf/cpsweek/ErnstADFMPYYZ19}, or to optimize some parameters.
For example, for the alphabet size, there should exist a trade-off between the computation cost and covering a larger class of signals. 
For the input mapper, we used a very simple input mapper in the explained as explained in \cref{section:experiments}.
Investigation of a good method to give an appropriate alphabets or an input mapper is future work.
It is also an interesting future work to use an optimization-based conformance testing of CPSs~\cite{DBLP:journals/tcps/RoehmOWA19} instead of our robustness-guided equivalence testing.
\begin{acks}
	\LongVersion{%
		This is the author (and extended) version of the manuscript of the same name published in the proceedings of the 23rd ACM International Conference on Hybrid Systems: Computation and Control (HSCC 2020).
	The final version is available at \url{dl.acm.org}.
	This version contains additional proofs.
	}%
        Thanks are due to Ichiro Hasuo for a useful feedback.
	This work is partially supported
	by JST ERATO HASUO Metamathematics for Systems Design Project (No.\ JPMJER1603) and
       by JSPS Grants-in-Aid No.\ 15KT0012 \& 18J22498.
\end{acks}

\newpage
	\newcommand{\LNCS}{Lecture Notes in Computer Science}

\bibliographystyle{ACM-Reference-Format}
\bibliography{dblp_refs}

\LongVersion{
\newpage
\appendix
\onecolumn
\section{Omitted proofs}

\subsection{Proof of \cref{theorem:sound_correct_robust_infinite}}

\begin{proof}
 [\cref{theorem:sound_correct_robust_infinite}]
 We prove by induction on the structure of $\stlFml$.

 When $\stlFml = \top$, we have
 $\robust{\top}{\signal}{k}) = +\infty > 0$ and
 $\satisfy{\signal}{k}{\top}$.

 When $\stlFml = y > c$, we have $\robust{\ovar > \constant}{\signal}{k} = \dval_k(\ovar) - \constant$.
 If we have $\robust{\ovar > \constant}{\signal}{k} > 0$, 
 we have $\dval_k(\ovar) - \constant > 0$, and 
 $\satisfy{\signal}{k}{\ovar > \constant}$ holds.
 If we have $\satisfy{\signal}{k}{\ovar > \constant}$,
 we have $\dval_k(\ovar) - \constant > 0$, and 
 $\robust{\ovar > \constant}{\signal}{k} \geq 0$ holds.

 When $\stlFml = \ovar < \constant$, we have $\robust{\ovar < \constant}{\signal}{k} = -\dval_k(\ovar) + \constant$.
 If we have $\robust{\ovar < \constant}{\signal}{k} > 0$, 
 we have $-\dval_k(\ovar) + \constant > 0$, and 
 $\satisfy{\signal}{k}{\ovar < \constant}$ holds.
 If we have $\satisfy{\signal}{k}{\ovar < \constant}$,
 we have $-\dval_k(\ovar) + \constant > 0$, and 
 $\robust{\ovar < \constant}{\signal}{k} \geq 0$ holds.

 When $\stlFml = \neg \stlFml'$, we have $\robust{\neg \stlFml'}{\signal}{k} = -\robust{\stlFml'}{\signal}{k}$. 
 If we have $\robust{\neg \stlFml'}{\signal}{k} > 0$,
 we have $\robust{\stlFml'}{\signal}{k} \leq 0$.
 Therefore, we have
 $\notsatisfy{\signal}{k}{\stlFml'}$ and 
 we have 
 $\satisfy{\signal}{k}{\neg\stlFml'}$.
 If we have $\satisfy{\signal}{k}{\neg\stlFml'}$, 
 we have 
 $\notsatisfy{\signal}{k}{\stlFml'}$.
 By induction hypothesis, we have
 $\satisfy{\signal}{k}{\stlFml'} < 0$
 and we have
 $\satisfy{\signal}{k}{\neg\stlFml'} \geq 0$.

 When $\stlFml = \stlFml' \lor \stlFml''$, we have 
 $\robust{\stlFml' \lor \stlFml''}{\signal}{k} = 
 \max\{
 \robust{\stlFml'}{\signal}{k},
 \robust{\stlFml''}{\signal}{k}\}$.
 If we have $\robust{\stlFml' \lor \stlFml'}{\signal}{k} > 0$,
 we have
 $\robust{\stlFml'}{\signal}{k} > 0$ or
 $\robust{\stlFml''}{\signal}{k} > 0$.
 By induction hypothesis, we have
 $\satisfy{\signal}{k}{\stlFml'}$ or
 $\satisfy{\signal}{k}{\stlFml''}$, and therefore,
 we have 
 $\satisfy{\signal}{k}{\stlFml' \lor \stlFml''}$.

 When $\stlFml = \NextOp \stlFml'$, we have
 $\robust{\NextOp\stlFml'}{\signal}{k} = \robust{\stlFml'}{\signal}{k+1}$.
 If we have $\robust{\NextOp \stlFml'}{\signal}{k} > 0$, 
 we have
 $\robust{\stlFml'}{\signal}{k+1} > 0$.
 By induction hypothesis, we have 
 $\satisfy{\signal}{k+1}{\stlFml'}$ and therefore, we have 
 $\satisfy{\signal}{k}{\NextOp \stlFml'}$.
 If we have 
 $\satisfy{\signal}{k}{\NextOp \stlFml'}$,
 we have
 $\satisfy{\signal}{k+1}{\stlFml'}$.
 By induction hypothesis, we have 
 $\robust{\stlFml'}{\signal}{k+1} \geq 0$, and therefore, we have
 $\robust{\NextOp \stlFml'}{\signal}{k} \geq 0$.

 When $\stlFml = \stlFml' \UntilOp{\interval} \stlFml''$,
 we have 
 $\robust{\stlFml \UntilOp{\interval} \stlFml'}{\signal}{k} = \sup_{l \in \interval[k + i,k + j]} \min\bigl(\robust{\stlFml'}{\signal}{l}, \min_{m \in \{k,k+1,\dots,l\}}\robust{\stlFml}{\signal}{m} \bigr)$.
 If we have $\robust{\stlFml \UntilOp{\interval} \stlFml'}{\signal}{k} > 0$, 
 there exists $l \in \interval[k + i,k + j]$ such that 
 we have
 $\robust{\stlFml'}{\signal}{l} > 0$ and 
 for any $m \in \{k,k+1,\dots,l\}$, we have
 $\robust{\stlFml}{\signal}{m} > 0$.
 By induction hypothesis, 
 there exists $l \in \interval[k + i,k + j]$ such that 
 we have
 $\satisfy{\signal}{l}{\stlFml'}$ and 
 for any $m \in \{k,k+1,\dots,l\}$, we have
 $\satisfy{\signal}{m}{\stlFml}$.
 Therefore, we have 
 $\satisfy{\signal}{k}{\stlFml \UntilOp{\interval} \stlFml'}$.
 If we have 
 $\satisfy{\signal}{k}{\stlFml \UntilOp{\interval} \stlFml'}$,
 there exists $l \in \interval[k + i,k + j]$ such that 
 we have
 $\satisfy{\signal}{l}{\stlFml'}$ and 
 for any $m \in \{k,k+1,\dots,l\}$, we have
 $\satisfy{\signal}{m}{\stlFml}$.
 By induction hypothesis, 
 there exists $l \in \interval[k + i,k + j]$ such that 
 we have
 $\robust{\stlFml'}{\signal}{l} \geq 0$ and 
 for any $m \in \{k,k+1,\dots,l\}$, we have
 $\robust{\stlFml}{\signal}{m} \geq 0$.
 Therefore,
 we have $\robust{\stlFml \UntilOp{\interval} \stlFml'}{\signal}{k} \geq 0$, 
\end{proof}

\subsection{Proof of \cref{theorem:sound_correct_robust_finite}}

First, we prove the following lemma.

\begin{lemma}
 \label{lemma:robust_neg}
 For an STL formula $\stlFml$ over $\OVar$, a \emph{finite} signal $\signalWithInside \in \FinOSignal$ over $\OVar$, and $k \in \N$ we have the following.  
 \begin{align*}
  \signal \not\in \infSem{\stlFml} &\iff \signal \in \supSem{\neg\stlFml}\\
  \signal \not\in \supSem{\stlFml} &\iff \signal \in \infSem{\neg\stlFml}
 \end{align*}
\end{lemma}
\begin{proof}
 The first part is proved as follows.
 \begin{align*}
  &\signal \not\in \infSem{\stlFml}\\ 
  \iff&
  \neg
  (\forall \signal'\in\InfOSignal.\,
  \signal \cdot \signal' \models \stlFml)\\
  \iff&
  \exists \signal'\in\InfOSignal.\,
  \signal \cdot \signal' \not\models \stlFml\\
  \iff&
  \exists \signal'\in\InfOSignal.\,
  \signal \cdot \signal' \models \neg\stlFml\\
  \iff&
  \signal \in \supSem{\neg\stlFml}
 \end{align*}

 The second part is proved as follows.
 \begin{align*}
  &\signal \not\in \supSem{\stlFml}\\ 
  \iff&
  \neg
  (\exists \signal'\in\InfOSignal.\,
  \signal \cdot \signal' \models \stlFml)\\
  \iff&
  \forall \signal'\in\InfOSignal.\,
  \signal \cdot \signal' \not\models \stlFml\\
  \iff&
  \forall \signal'\in\InfOSignal.\,
  \signal \cdot \signal' \models \neg\stlFml\\
  \iff&
  \signal \in \infSem{\neg\stlFml}
 \end{align*}
\end{proof}

\cref{theorem:sound_correct_robust_finite} is proved as follows.

\begin{proof}
 [\cref{theorem:sound_correct_robust_finite}]
 We prove by induction on the structure of $\stlFml$.

 When $\stlFml = \top$, we have
 $\sup(\RoSI{\top}{\signal}{k}) = \inf(\RoSI{\top}{\signal}{k}) = +\infty > 0$ and
 $\subseq{\signal}{k}{\abs{\signal} - 1} \in \FinOSignal = \supSem{\top} = \infSem{\top}$.

 When $\stlFml = y > c$, we have the following.
 \begin{align*}
  \sup(\RoSI{\ovar > \constant}{\signal}{k}) &=\sup_{\signal' \in \InfOSignal} \robust{\ovar > \constant}{\signal \cdot \signal'}{k} = 
  \begin{cases}
   \dval_k(\ovar) - \constant & \text{if $\abs{\signal} > k$}\\
   +\infty & \text{if $\abs{\signal} \leq k$}
  \end{cases}\\
  \inf(\RoSI{\ovar > \constant}{\signal}{k}) &=\inf_{\signal' \in \InfOSignal} \robust{\ovar > \constant}{\signal \cdot \signal'}{k} = 
  \begin{cases}
   \dval_k(\ovar) - \constant & \text{if $\abs{\signal} > k$}\\
   -\infty & \text{if $\abs{\signal} \leq k$}
  \end{cases}
 \end{align*}

 If we have $\sup(\RoSI{\ovar > \constant}{\signal}{k}) > 0$, we have $\dval_k(\ovar) > \constant$ or $\abs{\signal} \leq k$, and
 we have $\subseq{\signal}{k}{\abs{\signal}} \in \supSem{\ovar > \constant}$.

 If we have $\subseq{\signal}{k}{\abs{\signal}} \in \supSem{\ovar > \constant}$, 
 we have $\dval_k(\ovar) > \constant$ or $\abs{\signal} \leq k$, and thus,
 we have $\sup(\RoSI{\ovar > \constant}{\signal}{k}) \geq 0$.

 If we have $\inf(\RoSI{\ovar > \constant}{\signal}{k}) > 0$, we have $\abs{\signal} > k$ and $\dval_k(\ovar) > \constant$, and
 we have $\subseq{\signal}{k}{\abs{\signal}} \in \infSem{\ovar > \constant}$.

 If we have $\subseq{\signal}{k}{\abs{\signal}} \in \infSem{\ovar > \constant}$, 
 we have $\abs{\signal} > k$ and $\dval_k(\ovar) > \constant$, and thus,
 we have $\inf(\RoSI{\ovar > \constant}{\signal}{k}) \geq 0$.

  When $\stlFml = y < c$, we have the following.
 \begin{align*}
  \sup(\RoSI{\ovar < \constant}{\signal}{k}) &=\sup_{\signal' \in \InfOSignal} \robust{\ovar < \constant}{\signal \cdot \signal'}{k} = 
  \begin{cases}
   -\dval_k(\ovar) + \constant & \text{if $\abs{\signal} > k$}\\
   +\infty & \text{if $\abs{\signal} \leq k$}
  \end{cases}\\
  \inf(\RoSI{\ovar < \constant}{\signal}{k}) &=\inf_{\signal' \in \InfOSignal} \robust{\ovar < \constant}{\signal \cdot \signal'}{k} = 
  \begin{cases}
   -\dval_k(\ovar) + \constant & \text{if $\abs{\signal} > k$}\\
   -\infty & \text{if $\abs{\signal} \leq k$}
  \end{cases}
 \end{align*}

 If we have $\sup(\RoSI{\ovar < \constant}{\signal}{k}) > 0$, we have $\dval_k(\ovar) < \constant$ or $\abs{\signal} \leq k$, and
 we have $\subseq{\signal}{k}{\abs{\signal}} \in \supSem{\ovar < \constant}$.

 If we have $\subseq{\signal}{k}{\abs{\signal}} \in \supSem{\ovar < \constant}$, 
 we have $\dval_k(\ovar) < \constant$ or $\abs{\signal} \leq k$, and thus,
 we have $\sup(\RoSI{\ovar < \constant}{\signal}{k}) \geq 0$.

 If we have $\inf(\RoSI{\ovar < \constant}{\signal}{k}) > 0$, we have $\abs{\signal} > k$ and $\dval_k(\ovar) < \constant$, and
 we have $\subseq{\signal}{k}{\abs{\signal}} \in \infSem{\ovar < \constant}$.

 If we have $\subseq{\signal}{k}{\abs{\signal}} \in \infSem{\ovar < \constant}$, 
 we have $\abs{\signal} > k$ and $\dval_k(\ovar) < \constant$, and thus,
 we have $\inf(\RoSI{\ovar < \constant}{\signal}{k}) \geq 0$.

 When $\stlFml = \neg \stlFml'$, we have the following. 
 \begin{align*}
  \sup(\RoSI{\neg \stlFml'}{\signal}{k}) 
  =  \sup_{\signal' \in \InfOSignal} \robust{\neg\stlFml'}{\signal \cdot \signal}{k}
  &= \sup_{\signal' \in \InfOSignal} -\robust{\stlFml'}{\signal \cdot \signal}{k}\\
  &=  - \inf_{\signal' \in \InfOSignal} \robust{\stlFml'}{\signal \cdot \signal}{k}\\
  \inf(\RoSI{\neg \stlFml'}{\signal}{k}) 
  =  \inf_{\signal' \in \InfOSignal} \robust{\neg\stlFml'}{\signal \cdot \signal}{k}
  &= \inf_{\signal' \in \InfOSignal} -\robust{\stlFml'}{\signal \cdot \signal}{k}\\
  &=  - \sup_{\signal' \in \InfOSignal} \robust{\stlFml'}{\signal \cdot \signal}{k}
 \end{align*}

 If we have $\sup(\RoSI{\neg \stlFml'}{\signal}{k}) > 0$, 
 we have $\inf(\RoSI{\stlFml'}{\signal}{k}) < 0$.
 By induction hypothesis,
 we have 
 $\subseq{\signal}{k}{\abs{\signal}} \not\in \infSem{\stlFml'}$. 
 By \cref{lemma:robust_neg},
 we have $\subseq{\signal}{k}{\abs{\signal}} \in \supSem{\neg \stlFml'}$.

 If we have $\subseq{\signal}{k}{\abs{\signal}} \in \supSem{\neg \stlFml'}$,
 by \cref{lemma:robust_neg}, we have 
 $\subseq{\signal}{k}{\abs{\signal}} \not\in \infSem{\stlFml'}$. 
 By induction hypothesis,
 we have $\inf(\RoSI{\stlFml'}{\signal}{k}) < 0$, and thus, we have
 $\sup(\RoSI{\neg \stlFml'}{\signal}{k}) \geq 0$.

 If we have $\inf(\RoSI{\neg \stlFml'}{\signal}{k}) > 0$, 
 we have $\sup(\RoSI{\stlFml'}{\signal}{k}) < 0$.
 By induction hypothesis,
 we have 
 $\subseq{\signal}{k}{\abs{\signal}} \not\in \supSem{\stlFml'}$ and by \cref{lemma:robust_neg}.
 By \cref{lemma:robust_neg},
 we have $\subseq{\signal}{k}{\abs{\signal}} \in \infSem{\neg \stlFml'}$.

 If we have $\subseq{\signal}{k}{\abs{\signal}} \in \infSem{\neg \stlFml'}$,
 by \cref{lemma:robust_neg}, we have 
 $\subseq{\signal}{k}{\abs{\signal}} \not\in \supSem{\stlFml'}$. 
 By induction hypothesis, we have $\sup(\RoSI{\stlFml'}{\signal}{k}) < 0$, and thus, we have
 $\inf(\RoSI{\neg \stlFml'}{\signal}{k}) \geq 0$.

 When $\stlFml = \stlFml' \lor \stlFml''$, we have the following.
 \begin{align*}
  \sup(\RoSI{\stlFml' \lor \stlFml''}{\signal}{k}) 
  &= \sup_{\signal' \in \InfOSignal} \robust{\stlFml' \lor \stlFml''}{\signal\cdot \signal'}{k}\\
  &= \sup_{\signal' \in \InfOSignal} \max\bigl\{
  \robust{\stlFml' }{\signal\cdot \signal'}{k},
  \robust{\stlFml''}{\signal\cdot \signal'}{k}
  \bigr\}\\
  \inf(\RoSI{\stlFml' \lor \stlFml''}{\signal}{k}) 
  &= \inf_{\signal' \in \InfOSignal} \robust{\stlFml' \lor \stlFml''}{\signal\cdot \signal'}{k}\\
  &= \inf_{\signal' \in \InfOSignal} \max\bigl\{
  \robust{\stlFml' }{\signal\cdot \signal'}{k},
  \robust{\stlFml''}{\signal\cdot \signal'}{k}
  \bigr\}
 \end{align*}

 If we have $\sup(\RoSI{\stlFml' \lor \stlFml''}{\signal}{k}) > 0$, 
 there exists $\signal' \in \InfOSignal$ satisfying
 $\robust{\stlFml'}{\signal \cdot \signal'}{k} > 0$ or
 $\robust{\stlFml''}{\signal \cdot \signal'}{k} > 0$.
 By~\cref{theorem:sound_correct_robust_infinite}
 there exists $\signal' \in \InfOSignal$ satisfying
 $\satisfy{\signal \cdot \signal'}{k}{\stlFml'}$ or
 $\satisfy{\signal \cdot \signal'}{k}{\stlFml''}$, and thus,
 we have 
 $\subseq{\signal}{k}{\abs{\signal}} \in \supSem{\stlFml' \lor \stlFml''}$.

 If we have $\subseq{\signal}{k}{\abs{\signal}} \in \supSem{\stlFml' \lor \stlFml''}$,
 there exists $\signal' \in \InfOSignal$ satisfying
 $\satisfy{\signal \cdot \signal'}{k}{\stlFml' \lor \stlFml''}$.
 By~\cref{theorem:sound_correct_robust_infinite}
 there exists $\signal' \in \InfOSignal$ satisfying
 $\robust{\stlFml'}{\signal \cdot \signal'}{k} \geq 0$ or
 $\robust{\stlFml''}{\signal \cdot \signal'}{k} \geq 0$, and thus,
 we have $\sup(\RoSI{\stlFml' \lor \stlFml''}{\signal}{k}) \geq 0$.

 If we have $\inf(\RoSI{\stlFml' \lor \stlFml''}{\signal}{k}) > 0$, 
 for any $\signal' \in \InfOSignal$, we have
 $\robust{\stlFml'}{\signal \cdot \signal'}{k} > 0$ or
 $\robust{\stlFml''}{\signal \cdot \signal'}{k} > 0$.
 By~\cref{theorem:sound_correct_robust_infinite}
 for any $\signal' \in \InfOSignal$, we have
 $\satisfy{\signal \cdot \signal'}{k}{\stlFml'}$ or
 $\satisfy{\signal \cdot \signal'}{k}{\stlFml''}$.
 Therefore, we have 
 $\subseq{\signal}{k}{\abs{\signal}} \in \infSem{\stlFml' \lor \stlFml''}$.

 If we have $\subseq{\signal}{k}{\abs{\signal}} \in \infSem{\stlFml' \lor \stlFml''}$,
 for any $\signal' \in \InfOSignal$, we have
 $\satisfy{\signal \cdot \signal'}{k}{\stlFml' \lor \stlFml''}$.
 By~\cref{theorem:sound_correct_robust_infinite}
 for any $\signal' \in \InfOSignal$, we have
 $\robust{\stlFml' }{\signal \cdot \signal'}{k} \geq 0$ or 
 $\robust{\stlFml''}{\signal \cdot \signal'}{k} \geq 0$.
 Therefore, we have $\inf(\RoSI{\stlFml' \lor \stlFml''}{\signal}{k}) \geq 0$.

 When $\stlFml = \NextOp \stlFml'$, we have the following. 
 \begin{align*}
  \sup(\RoSI{\NextOp \stlFml'}{\signal}{k}) 
  =  \sup_{\signal' \in \InfOSignal} \robust{\NextOp \stlFml'}{\signal \cdot \signal}{k}
  &= \sup_{\signal' \in \InfOSignal} \robust{\stlFml'}{\signal \cdot \signal}{k+1}\\
  &= \sup(\RoSI{\stlFml'}{\signal}{k+1})\\
  \inf(\RoSI{\NextOp \stlFml'}{\signal}{k}) 
  =  \inf_{\signal' \in \InfOSignal} \robust{\NextOp \stlFml'}{\signal \cdot \signal}{k}
  &= \inf_{\signal' \in \InfOSignal} \robust{\stlFml'}{\signal \cdot \signal}{k+1}\\
  &= \inf(\RoSI{\stlFml'}{\signal}{k+1})
 \end{align*}

 If we have $\sup(\RoSI{\NextOp \stlFml'}{\signal}{k}) = \sup(\RoSI{\stlFml'}{\signal}{k+1}) > 0$, 
 we have
 $\subseq{\signal}{k+1}{\abs{\signal}} \in \supSem{\stlFml'}$ and 
 $\subseq{\signal}{k}{\abs{\signal}} \in \supSem{\NextOp \stlFml'}$.

 If we have $\subseq{\signal}{k}{\abs{\signal}} \in \supSem{\NextOp \stlFml'}$,
 we have $\subseq{\signal}{k+1}{\abs{\signal}} \in \supSem{\stlFml'}$,
 therefore, 
 we have
 $\sup(\RoSI{\NextOp \stlFml'}{\signal}{k}) = \sup(\RoSI{\stlFml'}{\signal}{k+1}) \geq 0$ and

 If we have $\inf(\RoSI{\NextOp \stlFml'}{\signal}{k}) = \inf(\RoSI{\stlFml'}{\signal}{k+1}) > 0$, 
 we have 
 $\subseq{\signal}{k+1}{\abs{\signal}} \in \infSem{\stlFml'}$ and 
 $\subseq{\signal}{k}{\abs{\signal}} \in \infSem{\NextOp \stlFml'}$.

 If we have $\subseq{\signal}{k}{\abs{\signal}} \in \infSem{\NextOp \stlFml'}$,
 we have $\subseq{\signal}{k+1}{\abs{\signal}} \in \infSem{\stlFml'}$,
 therefore, 
 we have
 $\inf(\RoSI{\NextOp \stlFml'}{\signal}{k}) = \inf(\RoSI{\stlFml'}{\signal}{k+1}) \geq 0$.

 When $\stlFml = \stlFml' \UntilOp{\interval} \stlFml''$, we have the following.
 \begin{align*}
  &\sup(\RoSI{\stlFml' \UntilOp{\interval} \stlFml''}{\signal}{k}) \\
  =& \sup_{\signal' \in \InfOSignal}\sup_{l \in \interval[k + i,k + j]} \min\bigl(\robust{\stlFml''}{\signal \cdot \signal'}{l}, \min_{m \in \{k,k+1,\dots,l\}}\robust{\stlFml'}{\signal \cdot \signal'}{m} \bigr)\\
  &\inf(\RoSI{\stlFml' \UntilOp{\interval} \stlFml''}{\signal}{k}) \\
  =& \inf_{\signal' \in \InfOSignal}\sup_{l \in \interval[k + i,k + j]} \min\bigl(\robust{\stlFml''}{\signal \cdot \signal'}{l}, \min_{m \in \{k,k+1,\dots,l\}}\robust{\stlFml'}{\signal \cdot \signal'}{m} \bigr)
 \end{align*}

 If we have $\sup(\RoSI{\stlFml' \UntilOp{\interval} \stlFml''}{\signal}{k}) > 0$, 
 there exist $\signal' \in \InfOSignal$ and $l \in \interval[k + i,k + j]$ such that
 we have 
 $\robust{\stlFml''}{\signal \cdot \signal'}{l} > 0$ and
 for any $m \in \{k,k+1,\dots,l\}$,
 we have $\robust{\stlFml'}{\signal \cdot \signal'}{m} > 0$.
 By~\cref{theorem:sound_correct_robust_infinite},
 there exist $\signal' \in \InfOSignal$ and $l \in \interval[k + i,k + j]$ such that
 we have 
 $\satisfy{\signal \cdot \signal'}{l}{\stlFml''}$ and
 for any $m \in \{k,k+1,\dots,l\}$,
 we have $\satisfy{\signal \cdot \signal'}{m}{\stlFml'}$.
 Therefore, 
 there exist $\signal' \in \InfOSignal$ and $l \in \interval[k + i,k + j]$ satisfying
 $\satisfy{\signal \cdot \signal'}{k}{\stlFml' \UntilOp{\interval} \stlFml''}$ and we have
 $\subseq{\signal}{k}{\abs{\signal}} \in \supSem{\stlFml' \UntilOp{\interval} \stlFml''}$

 If we have 
 $\subseq{\signal}{k}{\abs{\signal}} \in \supSem{\stlFml' \UntilOp{\interval} \stlFml''}$,
 there exist $\signal' \in \InfOSignal$ and $l \in \interval[k + i,k + j]$ satisfying
 $\satisfy{\signal \cdot \signal'}{k}{\stlFml' \UntilOp{\interval} \stlFml''}$ and we have
 $\subseq{\signal}{k}{\abs{\signal}} \in \supSem{\stlFml' \UntilOp{\interval} \stlFml''}$.
 Therefore,
 there exist $\signal' \in \InfOSignal$ and $l \in \interval[k + i,k + j]$ such that
 we have 
 $\satisfy{\signal \cdot \signal'}{l}{\stlFml''}$ and
 for any $m \in \{k,k+1,\dots,l\}$,
 we have $\satisfy{\signal \cdot \signal'}{m}{\stlFml'}$.
 By~\cref{theorem:sound_correct_robust_infinite},
 there exist $\signal' \in \InfOSignal$ and $l \in \interval[k + i,k + j]$ such that
 we have 
 $\robust{\stlFml''}{\signal \cdot \signal'}{l} \geq 0$ and
 for any $m \in \{k,k+1,\dots,l\}$,
 we have $\robust{\stlFml'}{\signal \cdot \signal'}{m} \geq 0$, and thus,
 we have $\sup(\RoSI{\stlFml' \UntilOp{\interval} \stlFml''}{\signal}{k}) \geq 0$.

 If we have $\inf(\RoSI{\stlFml' \UntilOp{\interval} \stlFml''}{\signal}{k}) > 0$, 
 for any $\signal' \in \InfOSignal$, there exists $l \in \interval[k + i,k + j]$ such that
 we have 
 $\robust{\stlFml''}{\signal \cdot \signal'}{l} > 0$ and
 for any $m \in \{k,k+1,\dots,l\}$,
 we have $\robust{\stlFml'}{\signal \cdot \signal'}{m} > 0$.
 By~\cref{theorem:sound_correct_robust_infinite},
 for any $\signal' \in \InfOSignal$, there exists $l \in \interval[k + i,k + j]$ such that
 we have 
 $\satisfy{\signal \cdot \signal'}{l}{\stlFml''}$ and
 for any $m \in \{k,k+1,\dots,l\}$,
 we have $\satisfy{\signal \cdot \signal'}{m}{\stlFml'}$.
 Therefore, 
 for any $\signal' \in \InfOSignal$, there exists $l \in \interval[k + i,k + j]$ satisfying
 $\satisfy{\signal \cdot \signal'}{k}{\stlFml' \UntilOp{\interval} \stlFml''}$ and we have
 $\subseq{\signal}{k}{\abs{\signal}} \in \infSem{\stlFml' \UntilOp{\interval} \stlFml''}$

 If we have 
 $\subseq{\signal}{k}{\abs{\signal}} \in \infSem{\stlFml' \UntilOp{\interval} \stlFml''}$,
 for any $\signal' \in \InfOSignal$, there exists $l \in \interval[k + i,k + j]$ satisfying
 $\satisfy{\signal \cdot \signal'}{k}{\stlFml' \UntilOp{\interval} \stlFml''}$
 Therefore,
 for any $\signal' \in \InfOSignal$, there exists $l \in \interval[k + i,k + j]$ such that
 we have 
 $\satisfy{\signal \cdot \signal'}{l}{\stlFml''}$ and
 for any $m \in \{k,k+1,\dots,l\}$,
 we have $\satisfy{\signal \cdot \signal'}{m}{\stlFml'}$.
 By~\cref{theorem:sound_correct_robust_infinite},
 for any $\signal' \in \InfOSignal$, there exists $l \in \interval[k + i,k + j]$ such that
 we have 
 $\robust{\stlFml''}{\signal \cdot \signal'}{l} \geq 0$ and
 for any $m \in \{k,k+1,\dots,l\}$,
 we have $\robust{\stlFml'}{\signal \cdot \signal'}{m} \geq 0$.
 Thus,
 we have $\inf(\RoSI{\stlFml' \UntilOp{\interval} \stlFml''}{\signal}{k}) \geq 0$.
\end{proof}

\subsection{Proof of \cref{theorem:robust_finite_computation}}

\begin{proof}
 [\cref{theorem:robust_finite_computation}]
 Since both $\finRobust{\stlFml}{\signal}{k}$ and $\RoSI{\stlFml}{\signal}{k}$ are nonempty closed intervals, we have 
 $\RoSI{\stlFml}{\signal}{k} \subseteq \finRobust{\stlFml}{\signal}{k}$ if and only if
 we have
 $\inf(\finRobust{\stlFml}{\signal}{k}) \leq \inf(\RoSI{\stlFml}{\signal}{k}) \leq
 \sup(\RoSI{\stlFml}{\signal}{k}) \leq \sup(\finRobust{\stlFml}{\signal}{k})$.
 We prove the theorem by induction on the structure of $\stlFml$.

 When $\stlFml = \top$, we have
 $\sup(\RoSI{\top}{\signal}{k}) = \inf(\RoSI{\top}{\signal}{k}) = +\infty$ and
 $\finRobust{\stlFml}{\signal}{k} = [+\infty,+\infty]$. 
 Therefore, we have $\finRobust{\top}{\signal}{k} = \RoSI{\top}{\signal}{k}$.

 When $\stlFml = y > c$, we have the following.
 \begin{align*}
  \sup(\RoSI{\ovar > \constant}{\signal}{k}) &=\sup_{\signal' \in \InfOSignal} \robust{\ovar > \constant}{\signal \cdot \signal'}{k} = 
  \begin{cases}
   \dval_k(\ovar) - \constant & \text{if $\abs{\signal} > k$}\\
   +\infty & \text{if $\abs{\signal} \leq k$}
  \end{cases}\\
  \inf(\RoSI{\ovar > \constant}{\signal}{k}) &=\inf_{\signal' \in \InfOSignal} \robust{\ovar > \constant}{\signal \cdot \signal'}{k} = 
  \begin{cases}
   \dval_k(\ovar) - \constant & \text{if $\abs{\signal} > k$}\\
   -\infty & \text{if $\abs{\signal} \leq k$}
  \end{cases}
 \end{align*}
 Therefore, we have $\finRobust{\ovar > \constant}{\signal}{k} = \RoSI{\ovar > \constant}{\signal}{k}$.

 When $\stlFml = \ovar < \constant$, we have the following.
 \begin{align*}
  \sup(\RoSI{\ovar < \constant}{\signal}{k}) &=\sup_{\signal' \in \InfOSignal} \robust{\ovar < \constant}{\signal \cdot \signal'}{k} = 
  \begin{cases}
   -\dval_k(\ovar) + \constant & \text{if $\abs{\signal} > k$}\\
   +\infty & \text{if $\abs{\signal} \leq k$}
  \end{cases}\\
  \inf(\RoSI{\ovar < \constant}{\signal}{k}) &=\inf_{\signal' \in \InfOSignal} \robust{\ovar < \constant}{\signal \cdot \signal'}{k} = 
  \begin{cases}
   -\dval_k(\ovar) + \constant & \text{if $\abs{\signal} > k$}\\
   -\infty & \text{if $\abs{\signal} \leq k$}
  \end{cases}
 \end{align*}
 Therefore, we have $\finRobust{\ovar < \constant}{\signal}{k} = \RoSI{\ovar < \constant}{\signal}{k}$.

 When $\stlFml = \neg \stlFml'$, we have the following. 
 \begin{align*}
  \sup(\RoSI{\neg \stlFml'}{\signal}{k}) 
  =  \sup_{\signal' \in \InfOSignal} \robust{\neg\stlFml'}{\signal \cdot \signal}{k}
  &= \sup_{\signal' \in \InfOSignal} -\robust{\stlFml'}{\signal \cdot \signal}{k}\\
  &=  - \inf_{\signal' \in \InfOSignal} \robust{\stlFml'}{\signal \cdot \signal}{k}\\
  &=  - \inf(\RoSI{\stlFml'}{\signal}{k})\\
  \inf(\RoSI{\neg \stlFml'}{\signal}{k}) 
  =  \inf_{\signal' \in \InfOSignal} \robust{\neg\stlFml'}{\signal \cdot \signal}{k}
  &= \inf_{\signal' \in \InfOSignal} -\robust{\stlFml'}{\signal \cdot \signal}{k}\\
  &=  - \sup_{\signal' \in \InfOSignal} \robust{\stlFml'}{\signal \cdot \signal}{k}\\
  &=  - \sup(\RoSI{\stlFml'}{\signal}{k})
 \end{align*}
 Therefore, we have the following.
 \begin{displaymath}
  \RoSI{\neg\stlFml'}{\signal}{k} = -\RoSI{\stlFml'}{\signal}{k} \subseteq 
  -\finRobust{\stlFml'}{\signal}{k} = \finRobust{\neg\stlFml'}{\signal}{k}
 \end{displaymath}

 When $\stlFml = \stlFml' \lor \stlFml''$, we have the following.
 \begin{align*}
  \sup(\RoSI{\stlFml' \lor \stlFml''}{\signal}{k}) 
  &= \sup_{\signal' \in \InfOSignal} \robust{\stlFml' \lor \stlFml''}{\signal\cdot \signal'}{k}\\
  &= \sup_{\signal' \in \InfOSignal} \max\bigl\{
  \robust{\stlFml' }{\signal\cdot \signal'}{k},
  \robust{\stlFml''}{\signal\cdot \signal'}{k}
  \bigr\}\\
 &= \max\bigl\{
 \sup_{\signal' \in \InfOSignal} \robust{\stlFml' }{\signal\cdot \signal'}{k},
 \sup_{\signal' \in \InfOSignal} \robust{\stlFml''}{\signal\cdot \signal'}{k}
 \bigr\}\\
 &= \max\bigl\{
 \sup(\RoSI{\stlFml' }{\signal}{k}),
 \sup(\RoSI{\stlFml''}{\signal}{k})
 \bigr\}\\
 &\leq \max\bigl\{
 \sup(\finRobust{\stlFml' }{\signal}{k}),
 \sup(\finRobust{\stlFml''}{\signal}{k})
 \bigr\}\\
  &=\sup(\finRobust{\stlFml' \lor \stlFml''}{\signal}{k})\\
  \inf(\RoSI{\stlFml' \lor \stlFml''}{\signal}{k}) 
  &= \inf_{\signal' \in \InfOSignal} \robust{\stlFml' \lor \stlFml''}{\signal\cdot \signal'}{k}\\
  &= \inf_{\signal' \in \InfOSignal} \max\bigl\{
  \robust{\stlFml' }{\signal\cdot \signal'}{k},
  \robust{\stlFml''}{\signal\cdot \signal'}{k}
  \bigr\}\\
 &\geq \max\bigl\{
 \inf_{\signal' \in \InfOSignal} \robust{\stlFml' }{\signal\cdot \signal'}{k},
 \inf_{\signal' \in \InfOSignal} \robust{\stlFml''}{\signal\cdot \signal'}{k}
 \bigr\}\\
 &= \max\bigl\{
 \inf(\RoSI{\stlFml' }{\signal}{k}),
 \inf(\RoSI{\stlFml''}{\signal}{k})
 \bigr\}\\
 &\geq \max\bigl\{
 \inf(\finRobust{\stlFml' }{\signal}{k}),
 \inf(\finRobust{\stlFml''}{\signal}{k})
 \bigr\}\\
 &= \inf(\finRobust{\stlFml' \lor \stlFml''}{\signal}{k})
 \end{align*}
 Therefore, we have
 $\RoSI{\stlFml' \lor \stlFml''}{\signal}{k} \subseteq \finRobust{\stlFml' \lor \stlFml''}{\signal}{k}$.

 When $\stlFml = \NextOp \stlFml'$, we have the following. 
 \begin{align*}
  \sup(\RoSI{\NextOp \stlFml'}{\signal}{k}) 
  =  \sup_{\signal' \in \InfOSignal} \robust{\NextOp \stlFml'}{\signal \cdot \signal}{k}
  &= \sup_{\signal' \in \InfOSignal} \robust{\stlFml'}{\signal \cdot \signal}{k+1}\\
  &= \sup(\RoSI{\stlFml'}{\signal}{k+1})\\
  \inf(\RoSI{\NextOp \stlFml'}{\signal}{k}) 
  =  \inf_{\signal' \in \InfOSignal} \robust{\NextOp \stlFml'}{\signal \cdot \signal}{k}
  &= \inf_{\signal' \in \InfOSignal} \robust{\stlFml'}{\signal \cdot \signal}{k+1}\\
  &= \inf(\RoSI{\stlFml'}{\signal}{k+1})
 \end{align*}
 Therefore, we have the following.
 \begin{displaymath}
  \RoSI{\NextOp\stlFml'}{\signal}{k} = \RoSI{\stlFml'}{\signal}{k+1} \subseteq 
  \finRobust{\stlFml'}{\signal}{k+1} = \finRobust{\NextOp\stlFml'}{\signal}{k}
 \end{displaymath}

 When $\stlFml = \stlFml' \UntilOp{\interval} \stlFml''$, we have the following.
 \begin{align*}
  &\sup(\RoSI{\stlFml' \UntilOp{\interval} \stlFml''}{\signal}{k}) \\
  =& \sup_{\signal' \in \InfOSignal}\sup_{l \in \interval[k + i,k + j]} \min\bigl(\robust{\stlFml''}{\signal \cdot \signal'}{l}, \min_{m \in \{k,k+1,\dots,l\}}\robust{\stlFml'}{\signal \cdot \signal'}{m} \bigr)\\
  =& \sup_{l \in \interval[k + i,k + j]} \sup_{\signal' \in \InfOSignal}\min\bigl(\robust{\stlFml''}{\signal \cdot \signal'}{l}, \min_{m \in \{k,k+1,\dots,l\}}\robust{\stlFml'}{\signal \cdot \signal'}{m} \bigr)\\
  \leq& \sup_{l \in \interval[k + i,k + j]} \min\bigl(\sup_{\signal' \in \InfOSignal}\robust{\stlFml''}{\signal \cdot \signal'}{l}, \sup_{\signal' \in \InfOSignal}\min_{m \in \{k,k+1,\dots,l\}}\robust{\stlFml'}{\signal \cdot \signal'}{m} \bigr)\\
  \leq& \sup_{l \in \interval[k + i,k + j]} \min\bigl(\sup_{\signal' \in \InfOSignal}\robust{\stlFml''}{\signal \cdot \signal'}{l}, \min_{m \in \{k,k+1,\dots,l\}}\sup_{\signal' \in \InfOSignal}\robust{\stlFml'}{\signal \cdot \signal'}{m} \bigr)\\
  =& \sup_{l \in \interval[k + i,k + j]} \min\bigl(\sup(\finRobust{\stlFml''}{\signal}{l}), \min_{m \in \{k,k+1,\dots,l\}}\sup(\finRobust{\stlFml'}{\signal}{m}) \bigr)\\
  =& \sup(\finRobust{\stlFml' \UntilOp{\interval} \stlFml''}{\signal}{k}) \\
  &\inf(\RoSI{\stlFml' \UntilOp{\interval} \stlFml''}{\signal}{k}) \\
  =& \inf_{\signal' \in \InfOSignal}\sup_{l \in \interval[k + i,k + j]} \min\bigl(\robust{\stlFml''}{\signal \cdot \signal'}{l}, \min_{m \in \{k,k+1,\dots,l\}}\robust{\stlFml'}{\signal \cdot \signal'}{m} \bigr)\\
  \geq& \sup_{l \in \interval[k + i,k + j]}\inf_{\signal' \in \InfOSignal} \min\bigl(\robust{\stlFml''}{\signal \cdot \signal'}{l}, \min_{m \in \{k,k+1,\dots,l\}}\robust{\stlFml'}{\signal \cdot \signal'}{m} \bigr)\\
  =& \sup_{l \in \interval[k + i,k + j]} \min\bigl(\inf_{\signal' \in \InfOSignal}\robust{\stlFml''}{\signal \cdot \signal'}{l}, \inf_{\signal' \in \InfOSignal}\min_{m \in \{k,k+1,\dots,l\}}\robust{\stlFml'}{\signal \cdot \signal'}{m} \bigr)\\
  =& \sup_{l \in \interval[k + i,k + j]} \min\bigl(\inf_{\signal' \in \InfOSignal}\robust{\stlFml''}{\signal \cdot \signal'}{l}, \min_{m \in \{k,k+1,\dots,l\}}\inf_{\signal' \in \InfOSignal}\robust{\stlFml'}{\signal \cdot \signal'}{m} \bigr)\\
  =& \sup_{l \in \interval[k + i,k + j]} \min\bigl(\inf(\finRobust{\stlFml''}{\signal}{l}), \min_{m \in \{k,k+1,\dots,l\}}\inf(\finRobust{\stlFml'}{\signal}{m}) \bigr)\\
  =& \inf(\finRobust{\stlFml' \UntilOp{\interval} \stlFml''}{\signal}{k})
 \end{align*}
 Therefore, we have $\RoSI{\stlFml' \UntilOp{\interval} \stlFml''}{\signal}{k} \subseteq \finRobust{\stlFml' \UntilOp{\interval} \stlFml''}{\signal}{k}$.
\end{proof}

\section{Omitted Experiment Result}
\cref{table:result_ratio} shows the ratio of the time to falsify as many properties as \breach{}.

\begin{table}[tb]
 \caption{The ratio of the time to falsify as many properties as \breach{}. There is no entry for $\varphi_4$ because the number of falsified properties by \breach{} was not constant. The cells with N/A show that the method could not falsify as many properties as \breach{}.}
 \label{table:result_ratio}
 \begin{tabular}[t]{c|c|c|c}
&\breach/\random&\breach/\hc&\breach/\ga\\\hline
$\varphi_{1}$ &N/A&N/A&0.174747\\
$\varphi_{2}$ &0.0657534&0.0468933&0.0426136\\
$\varphi_{3}$ &2.40107&1.76962&1.74259\\
$\varphi_{5}$ &3.32821&3.43027&2.70304\\
$\varphi_{6,\mathrm{tiny}}$ &1.39079&1.27869&0.880527\\
$\varphi_{6,\mathrm{small}}$ &1.48069&1.7071&1.38028\\
$\varphi_{6,\mathrm{medium}}$ &N/A&3.10013&1.90094\\
$\varphi_{6,\mathrm{large}}$ &N/A&3.33556&2.88035\\
$\varphi_{6,\mathrm{huge}}$ &5.2717&6.225&3.87649\\
$\varphi_{6,\mathrm{gigantic}}$ &6.4918&2.98888&4.56795\\
$\varphi_7$ &0.363057&0.36248&0.24333\\
 \end{tabular}
\end{table}}
\end{document}